\documentclass[aps,pra,reprint,twocolumn,nofootinbib]{revtex4-1}
\usepackage{amsfonts,amsmath,amsthm}
\usepackage{graphicx,xspace}

\newcommand{\QKD}{QKD\xspace}
\newcommand{\MDI}{MDI\xspace}
\newcommand{\MR}{MR\xspace}
\newcommand{\Xbasis}{{\mathtt X}\xspace}
\newcommand{\Zbasis}{{\mathtt Z}\xspace}
\newcommand{\Bbasis}{{\mathtt B}\xspace}
\newcommand{\even}{{\text{e}}}
\newcommand{\odd}{{\text{o}}}

\newtheorem{Thrm}{Theorem}
\newtheorem{Lem}{Lemma}
\newtheorem{Cor}{Corollary}
\newtheorem{Rem}{Remark}

\DeclareMathOperator{\Width}{Width}
\DeclareMathOperator{\BigOh}{O}

\begin{document}
\title{Security Of Finite-Key-Length Measurement-Device-Independent Quantum
 Key Distribution Using Arbitrary Number Of Decoys}

\author{H. F. Chau}
\thanks{email: \texttt{hfchau@hku.hk}}
\affiliation{Department of Physics, University of Hong Kong, Pokfulam Road,
 Hong Kong}
\date{\today}

\begin{abstract}
 In quantum key distribution, measurement-device-independent and decoy-state
 techniques enable the two cooperative agents to establish a shared secret key
 using imperfect measurement devices and weak Poissonian sources,
 respectively.
 Investigations so far are not comprehensive as they restrict to less than or
 equal to four decoy states.  Moreover, many of them involves pure numerical
 studies.
 Here I report a general security proof that works for any fixed number of
 decoy states and any fixed raw key length.
 The two key ideas involved here.
 The first one is the repeated application of the inversion formula for
 Vandermonde matrix to obtain various bounds on certain yields and error
 rates.
 The second one is the use of a recently proven generalization of the
 McDiarmid inequality.
 These techniques rise the best provably secure key rate of the
 measurement-device-independent version of the BB84 scheme by at least
 1.25~times and increase the workable distance between the two cooperative
 agents from slightly less than 60~km to slightly greater than 130~km in case
 there are $10^{10}$ photon pulse pair sent without a quantum repeater.
\end{abstract}

\maketitle

\section{Introduction}
\label{Sec:Intro}
 Quantum key distribution (\QKD) is the art for two trusted agents, commonly
 refers to as Alice and Bob, to share a provably secure secret key by
 preparing and measuring quantum states that are transmitted through a noisy
 channel controlled by an eavesdropper Eve who has unlimited computational
 power.
 In realistic \QKD setup, decoy-state technique allows Alice and Bob to obtain
 their secret key using the much more practical weak phase-randomized
 Poissonian sources~\cite{Wang05,LMC05}.  In addition,
 measurement-device-independent (\MDI) method enables them to use imperfect
 apparatus that may be controlled by Eve to perform measurement~\cite{LCQ12}.
 Decoy-state technique has been extensively studied.
 In fact, this technique can be applied to many different \QKD
 schemes~\cite{Wang05,LMC05,VV14,BMFBB16,DSB18}.
 Researches on the effective use of a general number of decoys have been
 conducted~\cite{Hayashi07,HN14,Chau18,CN20}.
 The effect of finite raw key length on the key rate has been
 investigated~\cite{HN14,LCWXZ14,BMFBB16,Chau18,CN20}.
 Nonetheless, security and efficiency analyses on the combined use of
 decoy-state and \MDI techniques are less comprehensive.  So far, they are
 restricted to less than or equal to four decoy
 states~\cite{MFR12,SGLL13,SGLL13e,CXCLTL14,XXL14,YZW15,ZYW16,ZZRM17,MZZZZW18,WBZJL19}.
 Furthermore, it is unclear how to extend these methods analytically to an
 arbitrary but fixed number of decoys.
 Along a slightly different line, the case of finite raw key length for the
 combined use of decoy-state and \MDI techniques has been studied.
 So far, these studies applied Azuma, Hoeffding and Sefling inequalities as
 well as Chernoff bound in a straightforward
 manner~\cite{CXCLTL14,YZW15,ZYW16,ZZRM17,MZZZZW18,WBZJL19}.

 Here I report the security analysis and a key rate formula for the
 BB84-based~\cite{BB84} \MDI-\QKD using passive partial Bell state detection
 for finite raw key length with the condition that Alice and Bob each uses an
 arbitrary but fixed number of decoys.
 One of the key ideas in this work is the repeated use of the analytical
 formula for the elements of the inverse of a Vandermonde matrix.  A tight
 bound on various yields and error rates for a general number of decoys can
 then be obtained through this analytical formula.  (Actually, Yuan
 \emph{et al.} also used repeated Vandermonde matrix inversion to obtain upper
 and lower bounds of the so-called two-single-photon yield in case one of the
 photon intensities used is $0$~\cite{YZLM16}.  Nevertheless, the bounds
 reported here are more general and powerful than theirs.)  The other key idea
 used here is the application of a powerful generalization of the McDiarmid
 inequality in mathematical statistics recently proven in Ref.~\cite{CN20}.
 This inequality is effective to tackle finite size statistical fluctuation of
 certain error rates involved in the key rate formula.

 I compute the secure key rate for the \MDI-version of the BB84 scheme using
 the setup and channel studied by Zhou \emph{et al.} in Ref.~\cite{ZYW16}.
 The best provably secure key rate for this setup before this work are
 reported by Mao \emph{et al.} in Ref.~\cite{MZZZZW18}.  Compared to their
 work, in case the total number of photon pulse pair send by Alice and Bob is
 $10^{10}$, the provably secure key rate using this new method is increased by
 at least 125\%.  Besides, the maximum transmission distance is increased from
 slightly less than 60~km to slightly greater than 130~km.
 This demonstrates the effectiveness of this new approach for \MDI-\QKD.
 
\section{The \MDI-\QKD Protocol}
\label{Sec:Protocol}
 In this paper, the polarization of all photon pulses are prepared either in
 $\Xbasis$-basis with photon intensity $\mu_{\Xbasis,i}$ (for $i=1,2,\cdots,
 k_\Xbasis$) or in $\Zbasis$-basis with photon intensity $\mu_{\Zbasis,i}$
 (for $i=1,2,\cdots,k_\Zbasis$).  For simplicity, I label these photon
 intensities in descending order by $\mu_{\Xbasis,1} > \mu_{\Xbasis,2} >
 \cdots > \mu_{\Xbasis,k_\Xbasis} \ge 0$ and similarly for
 $\mu_{\Zbasis,i}$'s.  I denote the probability of choosing the preparation
 basis $\Bbasis \in \{ \Xbasis,\Zbasis \}$ by $p_\Bbasis$ and the probability
 of choosing photon intensity $\mu_{\Bbasis,i}$ given the preparation basis
 $\Bbasis$ by $p_{i\mid\Bbasis}$.

 Here I study the following \MDI-\QKD protocol, which is a BB84-based scheme
 originally studied in Refs.~\cite{LCQ12,CXCLTL14}.
\begin{enumerate}
 \item Alice and Bob each has a phase-randomized Poissonian distributed
  source.  Each of them randomly and independently prepares a photon pulse and
  sends it to the untrusted third party Charlie.  They jot down the intensity
  and polarization used for each pulse.
  \label{Scheme:prepare}
 \item Charlie performs a partial Bell state measurement like the ones in
  Refs.~\cite{LCQ12,CXCLTL14,MR12}.  He publicly announces the measurement
  result including non-detection and inconclusive events.
  \label{Scheme:measure}
 \item Alice and Bob reveal the basis and intensity they used for each of
  their prepared photon pulse.  If the preparation bases of a pair of photon
  pulses they have sent to Charlie for Bell basis measurement disagree, they
  discard them.  If both pulses are prepared in the $\Xbasis$-basis, they
  reveal their preparation polarizations.  They also randomly reveal the
  preparation polarizations of a few pulses that they have both prepared in
  the $\Zbasis$-basis.  In this way, they can estimate the various yields and
  error rates to be defined in Sec.~\ref{Sec:Y_and_Q}.
  \label{Scheme:estimate}
 \item They use the preparation information of their remaining photon pulses
  that have been conclusively measured by Charlie to generate their raw secret
  keys and then perform error correction and privacy amplification on these
  keys to obtain their final secret keys according to the \MDI-\QKD procedure
  reported in Refs.~\cite{LCQ12,MR12}.  (Here I assume that Alice and Bob use
  forward reconciliation to establish the key.  The case of reverse
  reconciliation can be studied in a similar manner.)
  \label{Scheme:post_process}
\end{enumerate}


\section{Bounds On Various Yields And Error Rates In The \MDI-Setting}
\label{Sec:Y_and_Q}
 I use the symbol $Q_{\Bbasis,i,j}$ to denote the yield given that both Alice
 and Bob prepare their photons in $\Bbasis$-basis and that Alice (Bob) uses
 photon intensity $\mu_{\Bbasis,i}$ ($\mu_{\Bbasis,j}$) for $\Bbasis =
 \Xbasis, \Zbasis$ and $i,j = 1,2,\cdots,k_{\Bbasis}$.  More precisely, it is
 the portion of photon pairs prepared using the above description that Charlie
 declares conclusive detection.  Furthermore, I define the error rate of these
 photon pairs $E_{\Bbasis,i,j}$ as the portion of those conclusively detected
 photons above whose prepared polarizations by Alice and Bob are the same.
 And I set $\bar{E}_{\Bbasis,i,j} = 1 - E_{\Bbasis,i,j}$.
 Similar to the case of standard (that is, non-\MDI) implementation of \QKD,
 for phase randomized Poissonian photon sources~\cite{MR12},
\begin{equation}
 Q_{\Bbasis,i,j} = \sum_{a,b=0}^{+\infty} \frac{\mu_{\Bbasis,i}^a
 \mu_{\Bbasis,j}^b Y_{\Bbasis,a,b} \exp(-\mu_{\Bbasis,i})
 \exp(-\mu_{\Bbasis,j})}{a!\ b!}
 \label{E:Q_mu_def}
\end{equation}
 and
\begin{align}
 & Q_{\Bbasis,i,j} E_{\Bbasis,i,j} \nonumber \\
 ={} & \sum_{a,b=0}^{+\infty} \frac{\mu_{\Bbasis,i}^a \mu_{\Bbasis,j}^b
  Y_{\Bbasis,a,b} e_{\Bbasis,a,b} \exp(-\mu_{\Bbasis,i})
  \exp(-\mu_{\Bbasis,j})}{a!\ b!} .
 \label{E:E_mu_def}
\end{align}
 Here, $Y_{\Bbasis,a,b}$ is the probability of conclusive detection by Charlie
 given that the photon pulses sent by Alice (Bob) contains $a$ ($b$)~photons
 and $e_{\Bbasis,a,b}$ is the corresponding bit error rate of the raw key.
 Furthermore, I denote the yield conditioned on Alice preparing a vacuum state
 and Bob preparing in the $\Bbasis$-basis by the symbol
 $Y_{\Bbasis,0,\star}$.  Clearly, $Y_{\Bbasis,0,\star}$ obeys
\begin{equation}
 Y_{\Bbasis,0,\star} = \sum_{j=1}^{k_\Bbasis} p_{j\mid\Bbasis}
 \tilde{Y}_{\Bbasis,0,j} ,
 \label{E:Y_0*_relation}
\end{equation}
 where $\tilde{Y}_{\Bbasis,0,j}$ is the yield conditioned on Alice sending the
 vacuum state and Bob sending photon with intensity $\mu_{\Bbasis,j}$ in the
 $\Bbasis$-basis.

 I need to deduce the possible values of $Y_{\Bbasis,i,j}$'s and
 $Y_{\Bbasis,i,j} e_{\Bbasis,i,j}$ from Eqs.~\eqref{E:Q_mu_def}
 and~\eqref{E:E_mu_def}.  One way to do it is to compute various lower and
 upper bounds of $Y_{\Bbasis,i,j}$'s and $Y_{\Bbasis,i,j} e_{\Bbasis,i,j}$ by
 brute force optimization of truncated versions of Eqs.~\eqref{E:Q_mu_def}
 and~\eqref{E:E_mu_def} like the method reported in
 Refs.~\cite{MFR12,CXCLTL14,XXL14}.  However, this approach is rather
 inelegant and ineffective.
 Further note that Alice and Bob have no control on the values of
 $Y_{\Bbasis,a,b}$'s and $e_{\Bbasis,a,b}$'s since Charlie and Eve are not
 trustworthy.  All they know is that these variables are between 0 and 1.
 Fortunately, in the case of phase-randomized Poissonian distributed light
 source, Corollaries~\ref{Cor:Y0*} and~\ref{Cor:bounds_on_Yxx} in the Appendix
 can be used to bound $Y_{\Bbasis,0,\star}, Y_{\Bbasis,1,1}, Y_{\Bbasis,1,1}
 e_{\Bbasis,1,1}$ and $Y_{\Bbasis,1,1} \bar{e}_{\Bbasis,1,1}$ analytically,
 where $\bar{e}_{\Bbasis,1,1} \equiv 1 - e_{\Bbasis,1,1}$.  More importantly,
 these bounds are effective to analyze the key rate formula to be reported
 in Sec.~\ref{Sec:Rate}.
 Following the trick used in Refs.~\cite{Chau18,CN20}, by using the statistics
 of either all the $k_{\Bbasis}$ different photon intensities or all but the
 largest one used by Alice and Bob depending on the parity of $k_{\Bbasis}$,
 Corollaries~\ref{Cor:Y0*} and~\ref{Cor:bounds_on_Yxx} imply the following
 tight bounds
\begin{subequations}
\label{E:various_Y_and_e_bounds}
\begin{equation}
 Y_{\Bbasis,0,\star} \ge \sum_{i,j=1}^{k_\Bbasis} p_{j\mid\Bbasis}
 {\mathcal A}_{\Bbasis,0,i}^\even Q_{\Bbasis,i,j} ,
 \label{E:Y0*_bound}
\end{equation}
\begin{equation}
 Y_{\Bbasis,1,1} \ge Y_{\Bbasis,1,1}^\downarrow \equiv
 \sum_{i,j=1}^{k_\Bbasis} {\mathcal A}_{\Bbasis,1,i}^\odd
 {\mathcal A}_{\Bbasis,1,j}^\odd Q_{\Bbasis,i,j} - C_{\Bbasis,2}^2 ,
 \label{E:Y11_bound}
\end{equation}
\begin{equation}
 Y_{\Bbasis,1,1} e_{\Bbasis,1,1} \le \left( Y_{\Bbasis,1,1} e_{\Bbasis,1,1}
 \right)^\uparrow \equiv \sum_{i,j=1}^{k_\Bbasis}
 {\mathcal A}_{\Bbasis,1,i}^\even {\mathcal A}_{\Bbasis,1,j}^\even
 Q_{\Bbasis,i,j} E_{\Bbasis,i,j} ,
 \label{E:Ye11_bound}
\end{equation}
\begin{align}
 Y_{\Bbasis,1,1} e_{\Bbasis,1,1} &\ge{} \left( Y_{\Bbasis,1,1} e_{\Bbasis,1,1}
  \right)^\downarrow \nonumber \\
 &\equiv \sum_{i,j=1}^{k_\Bbasis} {\mathcal A}_{\Bbasis,1,i}^\odd
  {\mathcal A}_{\Bbasis,1,j}^\odd Q_{\Bbasis,i,j} E_{\Bbasis,i,j} -
  C_{\Bbasis,2}^2 ,
 \label{E:Ye11_special_bound}
\end{align}
 and
\begin{align}
 Y_{\Bbasis,1,1} \bar{e}_{\Bbasis,1,1} &\ge \left( Y_{\Bbasis,1,1}
  \bar{e}_{\Bbasis,1,1} \right)^\downarrow \nonumber \\
 &\equiv \sum_{i,j=1}^{k_\Bbasis} {\mathcal A}_{\Bbasis,1,i}^\odd
  {\mathcal A}_{\Bbasis,1,j}^\odd Q_{\Bbasis,i,j} \bar{E}_{\Bbasis,i,j} -
  C_{\Bbasis,2}^2
 \label{E:Y11bare1_bound}
\end{align}
\end{subequations}
 for $\Bbasis = \Xbasis, \Zbasis$.
 (The reason for using $\even$ and $\odd$ as superscripts is that it
 will be self-evident from the discussion below that for fixed $\Bbasis$ and
 $j$, there are even number of non-zero terms in
 ${\mathcal A}^\even_{\Bbasis,j,i}$ and odd number of non-zero terms in
 ${\mathcal A}^\odd_{\Bbasis,j,i}$.)
 For the above inequalities, in case $k_\Bbasis$ is even, then
\begin{subequations}
\label{E:various_A_ai}
\begin{equation}
 {\mathcal A}_{\Bbasis,j,i}^\even = A_j(\mu_{\Bbasis,i},\{ \mu_{\Bbasis,1},
 \mu_{\Bbasis,2}, \cdots, \mu_{\Bbasis,i-1}, \mu_{\Bbasis,i+1},\cdots,
 \mu_{\Bbasis,k_\Bbasis} \})
 \label{E:A_xi_even_def_for_even}
\end{equation}
 for $i=1,2,\cdots,k_\Bbasis$ and $j=0,1$.  Furthermore,
\begin{equation}
 {\mathcal A}_{\Bbasis,1,1}^\odd = 0
 \label{E:A_11_odd_def_for_even}
\end{equation}
 and
\begin{equation}
 {\mathcal A}_{\Bbasis,1,i}^\odd = A_1(\mu_{\Bbasis,i},\{ \mu_{\Bbasis,2},
 \mu_{\Bbasis,3},\cdots, \mu_{\Bbasis,i-1}, \mu_{\Bbasis,i+1},\cdots,
 \mu_{\Bbasis,k_\Bbasis} \})
 \label{E:A_1i_odd_def_for_even}
\end{equation}
 for $i=2,3,\cdots,k_\Bbasis$.  In addition,
\begin{widetext}
\begin{equation}
 C_{\Bbasis,2} = \left( \sum_{\ell=2}^{k_\Bbasis} \mu_{\Bbasis,2}
 \mu_{\Bbasis,3} \cdots \mu_{\Bbasis,\ell-1} \mu_{\Bbasis,\ell+1} \cdots
 \mu_{\Bbasis,k_\Bbasis} \right) \sum_{i=2}^{k_\Bbasis} \left\{
 \frac{1}{\mu_{\Bbasis,i} \prod_{t\ne 1,i} (\mu_{\Bbasis,i} -
 \mu_{\Bbasis,t})} \left[ \exp(\mu_{\Bbasis,i}) - \sum_{j=0}^{k_\Bbasis-2}
 \frac{\mu_{\Bbasis,i}^j}{j!} \right] \right\} .
 \label{E:C2_def_even}
\end{equation}
\end{widetext}
 Here I use the convention that the term involving $1/\mu_{\Bbasis,i}$ in the
 above summards with dummy index $i$ is equal to $0$ if $\mu_{\Bbasis,i} = 0$.

 Whereas in case $k_\Bbasis$ is odd, then 
\begin{equation}
 {\mathcal A}_{\Bbasis,1,i}^\odd = A_1(\mu_{\Bbasis,i},\{ \mu_{\Bbasis,1},
 \mu_{\Bbasis,2}, \cdots, \mu_{\Bbasis,i-1}, \mu_{\Bbasis,i+1},\cdots,
 \mu_{\Bbasis,k_\Bbasis} \})
 \label{E:A_1i_odd_def_for_odd}
\end{equation}
 for $i=1,2,\cdots,k_\Bbasis$.  Furthermore,
\begin{equation}
 {\mathcal A}_{\Bbasis,j,1}^\even = 0
 \label{E:A_x1_even_def_for_odd}
\end{equation}
 and
\begin{equation}
 {\mathcal A}_{\Bbasis,j,i}^\even = A_j(\mu_{\Bbasis,i},\{ \mu_{\Bbasis,2},
 \mu_{\Bbasis,3},\cdots, \mu_{\Bbasis,i-1}, \mu_{\Bbasis,i+1},\cdots,
 \mu_{\Bbasis,k_\Bbasis} \})
 \label{E:A_xi_even_def_for_odd}
\end{equation}
 for $j=1,2$ and $i=2,3,\cdots,k_\Bbasis$.  In addition,
\begin{widetext}
\begin{equation}
 C_{\Bbasis,2} = \left( \sum_{\ell=1}^{k_\Bbasis} \mu_{\Bbasis,1}
 \mu_{\Bbasis,2} \cdots \mu_{\Bbasis,\ell-1} \mu_{\Bbasis,\ell+1} \cdots
 \mu_{\Bbasis,k_\Bbasis} \right) \sum_{i=1}^{k_\Bbasis} \left\{
 \frac{1}{\mu_{\Bbasis,i} \prod_{t\ne i} (\mu_{\Bbasis,i} - \mu_{\Bbasis,t})}
 \left[ \exp(\mu_{\Bbasis,i}) - \sum_{j=0}^{k_\Bbasis-1}
 \frac{\mu_{\Bbasis,i}^j}{j!} \right] \right\} .
 \label{E:C2_def_odd}
\end{equation}
\end{widetext}
\end{subequations}

 Note that in Eq.~\eqref{E:various_A_ai},
\begin{subequations}
\label{E:generators_def}
\begin{equation}
 A_0(\mu,S) = \frac{\displaystyle -\exp(\mu) \prod_{s\in S}
 s}{\displaystyle \prod_{s\in S} (\mu - s)}
 \label{E:A_00_template}
\end{equation}
 and
\begin{equation}
 A_1(\mu,S) = \frac{\displaystyle -\exp(\mu) \sum_{s\in S} \left(
 \prod_{\substack{s'\in S \\ s' \ne s}} s' \right)}{\displaystyle
 \prod_{s\in S} (\mu - s)} .
 \label{E:A_11_template}
\end{equation}
\end{subequations}

 Note that different upper and lower bounds for $Y_{\Bbasis,1,1}$ have
 been obtained using similar Vandermonde matrix inversion technique in
 Ref.~\cite{YZLM16}.  The differences between those bounds and the actual
 value of $Y_{\Bbasis,1,1}$ depend on the yields $Y_{\Bbasis,i,j}$ with $i,j
 \ge 1$.  In contrast, the difference between the bound in
 Inequalities~\eqref{E:Y11_bound} and the actual value of $Y_{\Bbasis,1,1}$
 depend on $Y_{\Bbasis,i,j}$ with $i,j\ge k_\Bbasis$.  Thus,
 Inequality~\eqref{E:Y11_bound} and similarly also
 Inequality~\eqref{E:Ye11_bound} give more accurate estimates of
 $Y_{\Bbasis,1,1}$ and $Y_{\Bbasis,1,1} e_{\Bbasis,1,1}$, respectively.
 Furthermore, the bounds in Ref.~\cite{YZLM16} also work for the case of
 $\mu_{\Bbasis,k_{\Bbasis}} = 0$.  It is also not clear how to extend their
 method to bound for yields other than the two-single-photon events that are
 needed in computing the key rate for twin-field~\cite{LYDS18} and
 phase-matching~\cite{MZZ18} \MDI-\QKD{s}.

\section{The Key Rate Formula}
\label{Sec:Rate}
 The secure key rate $R$ is defined as the number of bits of secret key shared
 by Alice and Bob at the end of the protocol divided by the number of photon
 pulse pairs they have sent to Charlie.  In fact, the derivation of the key
 rate formula in Refs.~\cite{LCWXZ14,TL17,Chau18,CN20} for the case of
 standard \QKD can be easily modified to the case of \MDI-\QKD by making the
 following correspondences.
 (See also the key rate formula used in Ref.~\cite{MR12} for \MDI-\QKD.)
 The vacuum event in the standard \QKD is mapped to the event that both Alice
 and Bob send a vacuum photon pulse to Charlie.  The single photon event is
 mapped to the event that both Alice and Bob send a single photon to Charlie.
 The multiple photon event is mapped to the event that Alice and Bob are both
 sending neither a vacuum nor a single photon pulse to Charlie.  In the case
 of forward reconciliation, the result is
\begin{widetext}
\begin{equation}
 R \ge p_\Zbasis^2 \left\{ \langle \exp(-\mu) \rangle_{\Zbasis}
 Y_{\Zbasis,0,\star} + \langle \mu \exp(-\mu) \rangle_{\Zbasis}^2
 Y_{\Zbasis,1,1} [1-H_2(e_p)] - \Lambda_\text{EC} - \frac{\langle
 Q_{\Zbasis,i,j} \rangle_{i,j}}{\ell_\text{raw}} \left[ 6\log_2 \left(
 \frac{\chi}{\epsilon_\text{sec}} \right) + \log_2 \left(
 \frac{2}{\epsilon_\text{cor}} \right) \right] \right\} ,
 \label{E:key_rate_basic}
\end{equation}
\end{widetext}
 where $\langle f(\mu)\rangle_{\Zbasis} \equiv \sum_{i=1}^{k_\Zbasis}
 p_{i\mid\Zbasis} f(\mu_{\Zbasis,i})$, $\langle f(\Zbasis,i,j) \rangle_{i,j}
 \equiv \sum_{i,j=1}^{k_\Zbasis} p_{i\mid\Zbasis} p_{j\mid\Zbasis}
 f(\Zbasis,i,j)$, $H_2(x) \equiv -x \log_2 x - (1-x) \log_2 (1-x)$ is the
 binary entropy function, $e_p$ is the phase error rate of the single photon
 events in the raw key, and $\Lambda_\text{EC}$ is the actual number of bits
 of information that leaks to Eve as Alice and Bob perform error correction on
 their raw bits.  It is given by
\begin{equation}
 \Lambda_\text{EC} = f_\text{EC} \langle Q_{\Zbasis,i,j} H_2(E_{\Zbasis,i,j})
 \rangle_{i,j}
 \label{E:information_leakage}
\end{equation}
 where $f_\text{EC} \ge 1$ measures the inefficiency of the error-correcting
 code used.  In addition, $\ell_\text{raw}$ is the raw sifted key length
 measured in bits, $\epsilon_\text{cor}$ is the upper bound of the probability
 that the final keys shared between Alice and Bob are different, and
 $\epsilon_\text{sec} = (1- p_\text{abort}) \| \rho_\text{AE} - U_\text{A}
 \otimes \rho_\text{E} \|_1 / 2$.
 Here $p_\text{abort}$ is the chance that the scheme aborts without generating
 a key, $\rho_\text{AE}$ is the classical-quantum state describing the joint
 state of Alice and Eve, $U_\text{A}$ is the uniform mixture of all the
 possible raw keys created by Alice, $\rho_\text{E}$ is the reduced density
 matrix of Eve, and $\| \cdot \|_1$ is the trace
 norm~\cite{Renner05,KGR05,RGK05}.
 In other words, Eve has at most $\epsilon_\text{sec}$~bits of information on
 the final secret key shared by Alice and Bob.
 (In the literature, this is often referred to it as a
 $\epsilon_\text{cor}$-correct and $\epsilon_\text{sec}$-secure \QKD
 scheme~\cite{TL17}.)
 Last but not least, $\chi$ is a \QKD scheme specific factor which depends on
 the detailed security analysis used.
 In general, $\chi$ may also depend on other factors used in the \QKD scheme
 such as the number of photon intensities $k_{\Xbasis}$ and
 $k_{\Zbasis}$~\cite{LCWXZ14,Chau18,CN20}.

 In Inequality~\eqref{E:key_rate_basic}, the phase error of the raw key $e_p$
 obeys~\cite{Chau18,FMC10}
\begin{widetext}
\begin{equation}
 e_p \le e_{\Xbasis,1,1} + \bar{\gamma} \left(
 \frac{\epsilon_\text{sec}}{\chi}, e_{\Xbasis,1,1}, \frac{s_\Xbasis
 Y_{\Xbasis,1,1} \langle \mu \exp(-\mu) \rangle_{\Xbasis}^2}{\langle
 Q_{\Xbasis,i,j} \rangle_{i,j}}, \frac{s_\Zbasis Y_{\Zbasis,1,1} \langle \mu
 \exp(-\mu) \rangle_{\Zbasis}^2}{\langle Q_{\Zbasis,i,j} \rangle_{i,j}}
 \right)
 \label{E:e_p_bound}
\end{equation}
\end{widetext}
 with probability at least $1-\epsilon_\text{sec}/\chi$, where $\langle f(\mu)
 \rangle_\Xbasis \equiv \sum_{i=1}^{k_\Xbasis} p_{i\mid\Xbasis}
 f(\mu_{\Xbasis,i})$,
\begin{equation}
 \bar{\gamma}(a,b,c,d) \equiv \sqrt{\frac{(c+d)(1-b)b}{c d} \ \ln \left[
  \frac{c+d}{2\pi c d (1-b)b a^2} \right]} ,
 \label{E:gamma_def}
\end{equation}
 and $s_\Bbasis$ is the number of bits that are prepared and measured in
 $\Bbasis$ basis.
 Clearly, $s_\Zbasis = \ell_\text{raw}$ and $s_\Xbasis \approx p_\Xbasis^2
 s_\Zbasis \langle Q_{\Xbasis,i,j} \rangle_{i,j} / (p_\Zbasis^2 \langle
 Q_{\Zbasis,i,j} \rangle_{i,j})$.
 I also remark that $\bar{\gamma}$ becomes complex if $a,c,d$ are too large.
 This is because in this case no $e_p \ge e_{\Xbasis,1,1}$ exists with failure
 probability $a$.  In this work, all parameters are carefully picked so that
 $\bar{\gamma}$ is real.

 There are two ways to proceed.  The most general way is to directly find a
 lower bound for $Y_{\Zbasis,1,1}$.
 Specifically, by substituting Inequalities~\eqref{E:Y0*_bound},
 \eqref{E:Y11_bound} and~\eqref{E:e_p_bound} into
 Inequality~\eqref{E:key_rate_basic}, I obtain the following lower bound of
 the key rate
\begin{align}
 R &\ge \sum_{i,j=1}^{k_\Zbasis} {\mathcal B}_{\Zbasis,i,j} Q_{\Zbasis,i,j} -
  p_\Zbasis^2 \left\{ \vphantom{\frac{\chi}{\epsilon_\text{sec}}} \langle \mu
  \exp(-\mu) \rangle_\Zbasis^2 C_{\Zbasis,2}^2 [1 - H_2(e_p)] \right.
  \nonumber \\
 &\quad \left. + \Lambda_\text{EC} + \frac{\langle Q_{\Zbasis,i,j}
  \rangle_{i,j}}{\ell_\text{raw}} \left[ 6\log_2 \left(
  \frac{\chi}{\epsilon_\text{sec}} \right) + \log_2 \left(
  \frac{2}{\epsilon_\text{cor}} \right) \right] \right\} ,
 \label{E:key_rate_asym}
\end{align}
 where
\begin{align}
 {\mathcal B}_{\Zbasis,i,j} &= p_\Zbasis^2 \left\{ \langle \exp(-\mu)
  \rangle_{\Zbasis} {\mathcal A}_{\Zbasis,0,i}^\even p_{j\mid\Zbasis} \right.
  \nonumber \\
 &\quad \left. + \langle \mu \exp(-\mu) \rangle_{\Zbasis}^2
  {\mathcal A}_{\Zbasis,1,i}^\odd {\mathcal A}_{\Zbasis,1,j}^\odd [1-H_2(e_p)]
  \right\} .
 \label{E:b_n_def}
\end{align}

 Here I would like to point out that unlike the corresponding key rate
 formulae for standard \QKD in Refs.~\cite{Hayashi07,LCWXZ14,Chau18,CN20}, a
 distinctive feature of the key rate formula for \MDI-\QKD in
 Eq.~\eqref{E:key_rate_asym} is the presence of the $C_{\Zbasis,2}^2$ term.
 From Eq.~\eqref{E:various_A_ai}, provided that $\mu_{\Zbasis,i} -
 \mu_{\Zbasis,i+1}$ are all greater than a fixed positive number, the value of
 $C_{\Zbasis,2}^2$ decreases with $k_\Zbasis$.  This is the reason why the
 \MDI version of a \QKD scheme may require more decoys to attain a key rate
 comparable to the corresponding standard \QKD scheme.

 There is an alternative way to obtain the key rate formula discovered by
 Zhou \emph{et al.}~\cite{ZYW16} that works for BB84~\cite{BB84} and the
 six-state scheme~\cite{B98}.  Suppose the photon pulses prepared by Alice and
 Bob in Step~\ref{Scheme:prepare} of the \MDI-\QKD protocol in
 Sec.~\ref{Sec:Protocol} both contain a single photon.  Suppose further that
 they are prepared in the same basis.  Then, from Charlie and Eve's point of
 view, this two-single-photon state are the same irrespective of their
 preparation basis.  Consequently, $Y_{\Xbasis,1,1} = Y_{\Zbasis,1,1}$ (even
 though $e_{\Xbasis,1,1}$ need not equal $e_{\Zbasis,1,1}$).  That is to say,
 the secure key rate in Inequality~\eqref{E:key_rate_basic} also holds if
 $Y_{\Zbasis,1,1}$ there is replaced by $Y_{\Xbasis,1,1}$.  (Here I stress
 that the key generation basis is still $\Zbasis$.  But as $Y_{\Xbasis,1,1} =
 Y_{\Zbasis,1,1}$, I could use the bound on $Y_{\Xbasis,1,1}$ to obtain an
 alternative key rate formula for the same \MDI-\QKD scheme.)
 Following the same procedure above, I get
\begin{align}
 R &\ge \sum_{i,j=1}^{k_\Zbasis} {\mathcal B}'_{\Zbasis,i,j} Q_{\Zbasis,i,j} +
  \sum_{i,j=1}^{k_\Xbasis} {\mathcal B}_{\Xbasis,i,j} Q_{\Xbasis,i,j}
  \nonumber \\
 &\quad - p_\Zbasis^2 \left\{ \vphantom{\frac{\chi}{\epsilon_\text{sec}}}
  \langle \mu \exp(-\mu) \rangle_\Zbasis^2 C_{\Xbasis,2}^2 [1 - H_2(e_p)] +
  \Lambda_\text{EC} \right. \nonumber \\
 &\quad \left. + \frac{\langle Q_{\Zbasis,i,j} \rangle_{i,j}}{\ell_\text{raw}}
  \left[ 6\log_2 \left( \frac{\chi}{\epsilon_\text{sec}} \right) + \log_2
  \left( \frac{2}{\epsilon_\text{cor}} \right) \right] \right\} ,
 \label{E:key_rate_asym_alt}
\end{align}
 where
\begin{subequations}
\begin{equation}
 {\mathcal B}'_{\Zbasis,i,j} = p_\Zbasis^2 \langle \exp(-\mu)
 \rangle_{\Zbasis} {\mathcal A}_{\Zbasis,0,i}^\even p_{j\mid\Zbasis}
 \label{E:bprime_n_def}
\end{equation}
 and
\begin{equation}
 {\mathcal B}_{\Xbasis,i,j} = p_\Zbasis^2 \langle \mu \exp(-\mu)
 \rangle_{\Zbasis}^2 {\mathcal A}_{\Xbasis,1,i}^\odd
 {\mathcal A}_{\Xbasis,1,j}^\odd [1-H_2(e_p)] .
 \label{E:bX_n_def}
\end{equation}
\end{subequations}
 
\section{Treatments Of Phase Error And Statistical Fluctuation Due To Finite
 Raw Key Length On The Secure Key Rate}
\label{Sec:Finite_Size}
 In order to compute the lower bound on the key rate $R$ in
 Inequalities~\eqref{E:key_rate_asym} and~\eqref{E:key_rate_asym_alt}, I need
 to know the value of $e_{\Xbasis,1,1}$ through the
 Inequality~\eqref{E:e_p_bound}.  More importantly, I need to take into
 consideration the effects of finite raw key length on the key rate $R$ due to
 the statistical fluctuations in $e_{\Xbasis,1,1}$ and $Q_{\Zbasis,i,j}$'s.
 Here I do so by means of a McDiarmid-type of inequality in statistics first
 proven in Refs.~\cite{McDiarmid,McDiarmid1} and recently extended in
 Ref.~\cite{CN20}.

 Fluctuation of the first term in the R.H.S. of
 Inequality~\eqref{E:key_rate_asym}
 due to finite raw key length can be handled by Hoeffding inequality for
 hypergeometrically distributed random
 variables~\cite{Hoeffding,LCWXZ14,Chau18,CN20}, which is a special case of
 the McDiarmid inequality.  Using the technique reported in
 Refs.~\cite{Chau18,CN20}, the first term in the R.H.S. of
 Inequality~\eqref{E:key_rate_asym} can be regarded as a sum of $s_\Zbasis$
 hypergeometrically distributed random variables each taking on values from
 the set $\{ \langle Q_{\Zbasis,i,j} \rangle_{i,j} {\mathcal B}_{\Zbasis,i,j}
 / (p_{i\mid\Zbasis} p_{j\mid\Zbasis}) \}_{i,j=1}^{k_\Zbasis}$.  Using
 Hoeffding inequality for hypergeometrically distributed random
 variables~\cite{Hoeffding}, I conclude that the measured value of $\sum_{i,j}
 {\mathcal B}_{ij} Q_{\Zbasis,i,j}$ minus its actual value is greater than
 $\langle Q_{\Zbasis,i,j} \rangle_{i,j} \left[ \frac{\ln
 (\chi/\epsilon_\text{sec})}{2 s_\Zbasis} \right]^{1/2} \Width \left( \left\{
 \frac{{\mathcal B}_{\Zbasis,i,j}}{p_{i\mid\Zbasis} p_{j\mid\Zbasis}}
 \right\}_{i,j=1}^{k_\Zbasis} \right)$ with probability at most
 $\epsilon_\text{sec} / \chi$, where $\Width$ of a finite set of real numbers
 $S$ is defined as $\max S - \min S$.

 The value of $e_{\Xbasis,1,1}$ in the finite sampling size situation is more
 involved.  Here I adapt the recent results in Ref.~\cite{CN20} to give four
 upper bounds on $e_{\Xbasis,1,1}$.  Surely, I pick the best upper bound out
 of these four in the key rate analysis.
 The first step is to use the equality
\begin{subequations}
\label{E:e_Z11_identities}
\begin{align}
 e_{\Xbasis,1,1} &= \frac{Y_{\Xbasis,1,1} e_{\Xbasis,1,1}}{Y_{\Xbasis,1,1}}
  \label{E:e_Z11_identity1} \\
 &= \frac{Y_{\Xbasis,1,1} e_{\Xbasis,1,1}}{Y_{\Xbasis,1,1} e_{\Xbasis,1,1} +
  Y_{\Xbasis,1,1} \bar{e}_{\Xbasis,1,1}} .
  \label{E:e_Z11_identity2}
\end{align}
\end{subequations}
 To get the first two upper bounds of $e_{\Xbasis,1,1}$, I follow
 Ref.~\cite{CN20} by using Inequalities~\eqref{E:Y11_bound},
 \eqref{E:Ye11_bound} and~\eqref{E:Y11bare1_bound} together with applying
 Hoeffding inequality for hypergeometrically distributed random variables to
 study the statistical fluctuations of $\sum_{i,j=1}^{k_\Xbasis}
 {\mathcal A}_{\Xbasis,1,i}^\even {\mathcal A}_{\Xbasis,1,j}^\even
 Q_{\Xbasis,i,j} E_{\Xbasis,i,j}$, $\sum_{i,j=1}^{k_\Xbasis}
 {\mathcal A}_{\Xbasis,1,i}^\odd {\mathcal A}_{\Xbasis,1,j}^\odd
 Q_{\Xbasis,i,j}$ and $\sum_{i,j=1}^{k_\Xbasis}
 {\mathcal A}_{\Xbasis,1,i}^\odd {\mathcal A}_{\Xbasis,1,j}^\odd
 Q_{\Xbasis,i,j} \bar{E}_{\Xbasis,i,j}$.  The result is
\begin{equation}
 e_{\Xbasis,1,1} \le \frac{\left( Y_{\Xbasis,1,1} e_{\Xbasis,1,1}
 \right)^\uparrow + \Delta Y_{\Xbasis,1,1}
 e_{\Xbasis,1,1}}{Y_{\Xbasis,1,1}^\downarrow - \Delta Y_{\Xbasis,1,1}}
 \label{E:e_Z11_bound1}
\end{equation}
 and
\begin{align}
 e_{\Xbasis,1,1} &\le \left[ \left( Y_{\Xbasis,1,1} e_{\Xbasis,1,1}
  \right)^\uparrow + \Delta Y_{\Xbasis,1,1} e_{\Xbasis,1,1} \right] \left[
  \left( Y_{\Xbasis,1,1} e_{\Xbasis,1,1} \right)^\uparrow \right. \nonumber \\
 &\quad \left. + \left( Y_{\Xbasis,1,1} \bar{e}_{\Xbasis,1,1}
  \right)^\downarrow + \Delta Y_{\Xbasis,1,1} e_{\Xbasis,1,1} - \Delta
  Y_{\Xbasis,1,1} \bar{e}_{\Xbasis,1,1} \right]^{-1}
 \label{E:e_Z11_bound2}
\end{align}
 each with probability at least $1-2\epsilon_\text{sec}/\chi$, where
\begin{subequations}
\label{E:finite-size_Ye_s}
\begin{align}
 \Delta Y_{\Xbasis,1,1} e_{\Xbasis,1,1} &= \left[ \frac{\langle
  Q_{\Xbasis,i,j} \rangle_{i,j} \langle Q_{\Xbasis,i,j} E_{\Xbasis,i,j}
  \rangle_{i,j} \ln (\chi/\epsilon_\text{sec})}{2 s_\Xbasis} \right]^{1/2}
  \times \nonumber \\
 &\qquad \Width \left( \left\{ \frac{{\mathcal A}_{\Xbasis,1,i}^\even
  {\mathcal A}_{\Xbasis,1,j}^\even}{p_{i\mid\Xbasis} p_{j\mid\Xbasis}}
  \right\}_{i,j=1}^{k_\Xbasis} \right) ,
 \label{E:finite-size_Ye}
\end{align}
\begin{align}
 \Delta Y_{\Xbasis,1,1} &= \langle Q_{\Xbasis,i,j} \rangle_{i,j} \left[
  \frac{\ln (\chi/\epsilon_\text{sec})}{2 s_\Xbasis} \right]^{1/2} \times
  \nonumber \\
 &\qquad \Width \left( \left\{ \frac{{\mathcal A}_{\Xbasis,1,i}^\odd
  {\mathcal A}_{\Xbasis,1,j}^\odd}{p_{i\mid\Xbasis} p_{j\mid\Xbasis}}
  \right\}_{i,j=1}^{k_\Xbasis} \right)
 \label{E:finite-size_Y}
\end{align}
 and
\begin{align}
 \Delta Y_{\Xbasis,1,1} \bar{e}_{\Xbasis,1,1} &= \left[
  \frac{\langle Q_{\Xbasis,i,j} \rangle_{i,j} \langle Q_{\Xbasis,i,j}
  \bar{E}_{\Xbasis,i,j} \rangle_{i,j} \ln (\chi/\epsilon_\text{sec})}{2
  s_\Xbasis} \right]^{1/2} \times \nonumber \\
 &\qquad \Width \left( \left\{ \frac{{\mathcal A}_{\Xbasis,1,i}^\odd
  {\mathcal A}_{\Xbasis,1,j}^\odd}{p_{i\mid\Xbasis} p_{j\mid\Xbasis}}
  \right\}_{i,j=1}^{k_\Xbasis} \right) .
 \label{E:finite-size_Yebar}
\end{align}
\end{subequations}
 Note that in the above equations, $\langle f(\Xbasis,i,j) \rangle_{i,j}
 \equiv \sum_{i,j=1}^{k_\Xbasis} p_{i\mid\Xbasis} p_{j\mid\Xbasis}
 f(\Xbasis,i,j)$.

 Both the third and the fourth bounds of $e_{\Xbasis,1,1}$ use
 Eq.~\eqref{E:e_Z11_identity2}, Inequality~\eqref{E:Y11bare1_bound} and the
 modified McDiarmid inequality in Ref.~\cite{CN20}.  For the third one, the
 result is
\begin{align}
 e_{\Xbasis,1,1} &\le \frac{\left( Y_{\Xbasis,1,1} e_{\Xbasis,1,1}
  \right)^\uparrow}{\left( Y_{\Xbasis,1,1} e_{\Xbasis,1,1} \right)^\uparrow +
  \left( Y_{\Xbasis,1,1} \bar{e}_{\Xbasis,1,1} \right)^\downarrow - \Delta
  Y_{\Xbasis,1,1} \bar{e}_{\Xbasis,1,1}} \nonumber \\
 &\quad + \Delta e_{\Xbasis,1,1}
 \label{E:e_Z11_bound3}
\end{align}
 with probability at least $1 - 2\epsilon_\text{sec} / \chi$, where
\begin{widetext}
\begin{align}
 & \Delta e_{\Xbasis,1,1} \nonumber \\
 ={} & \left[ \frac{\langle Q_{\Xbasis,i,j}
  \rangle_{i,j} \langle Q_{\Xbasis,i,j} E_{\Xbasis,i,j} \rangle_{i,j} \ln
  (\chi/\epsilon_\text{sec})}{2 s_\Xbasis} \right]^{1/2} \left[ \left(
  Y_{\Xbasis,1,1} \bar{e}_{\Xbasis,1,1} \right)^\downarrow - \Delta
  Y_{\Xbasis,1,1} \bar{e}_{\Xbasis,1,1} \right] \Width \left( \left\{
  \frac{{\mathcal A}_{\Xbasis,1,i}^\even
  {\mathcal A}_{\Xbasis,1,j}^\even}{p_{i\mid\Xbasis} p_{j\mid\Xbasis}}
  \right\}_{i,j=1}^{k_\Xbasis} \right) \nonumber \\
 & \quad \times \left[ \left( Y_{\Xbasis,1,1} \bar{e}_{\Xbasis,1,1}
  \right)^\downarrow - \Delta Y_{\Xbasis,1,1} \bar{e}_{\Xbasis,1,1} + \left(
  Y_{\Xbasis,1,1} e_{\Xbasis,1,1} \right)^\downarrow \left( 1 -
  \frac{\langle Q_{\Xbasis,i,j} \rangle_{i,j}}{s_\Xbasis \langle
  Q_{\Xbasis,i,j} E_{\Xbasis,i,j} \rangle_{i,j}} \right) + \frac{\langle
  Q_{\Xbasis,i,j} \rangle_{i,j}^2}{s_\Xbasis^2 \langle Q_{\Xbasis,i,j}
  E_{\Xbasis,i,j} \rangle_{i,j}} \max_{i,j=1}^{k_\Xbasis} \left\{
  \frac{{\mathcal A}_{\Xbasis,1,i}^\even
  {\mathcal A}_{\Xbasis,1,j}^\even}{p_{i\mid\Xbasis} p_{j\mid\Xbasis}}
  \right\} \right]^{-1} \nonumber
  \displaybreak[1]
  \\
 & \quad \times \left[ \left( Y_{\Xbasis,1,1} \bar{e}_{\Xbasis,1,1}
  \right)^\downarrow - \Delta Y_{\Xbasis,1,1} \bar{e}_{\Xbasis,1,1} + \left(
  Y_{\Xbasis,1,1} e_{\Xbasis,1,1} \right)^\downarrow \left( 1 -
  \frac{\langle Q_{\Xbasis,i,j} \rangle_{i,j}}{s_\Xbasis \langle
  Q_{\Xbasis,i,j} E_{\Xbasis,i,j} \rangle_{i,j}} \right) + \frac{\langle
  Q_{\Xbasis,i,j} \rangle_{i,j}^2}{s_\Xbasis^2 \langle Q_{\Xbasis,i,j}
  E_{\Xbasis,i,j} \rangle_{i,j}} \min_{i,j=1}^{k_\Xbasis} \left\{
  \frac{{\mathcal A}_{\Xbasis,1,i}^\even
  {\mathcal A}_{\Xbasis,1,j}^\even}{p_{i\mid\Xbasis} p_{j\mid\Xbasis}}
  \right\} \right]^{-1} .
 \label{E:finite-size_e_Z11_shift3}
\end{align}
 And the fourth bound is
\begin{equation}
 e_{\Xbasis,1,1} \le \frac{\left( Y_{\Xbasis,1,1} e_{\Xbasis,1,1}
 \right)^\uparrow}{\left( Y_{\Xbasis,1,1} e_{\Xbasis,1,1} \right)^\uparrow +
 \left( Y_{\Xbasis,1,1} \bar{e}_{\Xbasis,1,1} \right)^\downarrow - \Delta
 Y_{\Xbasis,1,1} \bar{e}_{\Xbasis,1,1}} + \hat{r} \left[ \frac{\ln
 (\chi/\epsilon_\text{sec})}{2} \right]^{1/2}
 \label{E:e_Z11_bound4}
\end{equation}
 with probability at least $1 - 3\epsilon_\text{sec} / \chi$, where
\begin{align}
 \hat{r}^2 &\approx y^2 \sum_{m=1}^{k_\Xbasis^2} \frac{1}{w^{(m)}-x} \left( -
  \frac{1}{y+(t-\sum_{i<m} n^{(i)}+1)x+ \min {\mathcal W} + \sum_{i<m}
  n^{(i)} w^{(i)} + \mu[ w^{(m)}-x]} \right. \nonumber \\
 & \qquad - \frac{1}{y+(t-\sum_{i<m} n^{(i)}+1)x+\max {\mathcal W} +
  \sum_{i<m} n^{(i)} w^{(i)} + \mu[ w^{(m)}-x]} \nonumber \\
 & \qquad \left. \left. + \frac{2}{\Width({\mathcal W})} \ln \left\{
  \frac{y+(t-\sum_{i<m} n^{(i)}+1)x+\max {\mathcal W} + \sum_{i<m}
  n^{(i)} w^{(i)} + \mu [w^{(m)}-x]}{y+(t-\sum_{i<m} n^{(i)}+1)x+\min
  {\mathcal W} + \sum_{i<m} n^{(i)} w^{(i)} + \mu [w^{(m)}-x]} \right\}
 \right) \right|_{\mu = 0}^{n^{(m)}} .
 \label{E:finite-size_e_Z11_shift4}
\end{align}
\end{widetext}
 In the above equation,
\begin{subequations}
\begin{equation}
 y = \left( Y_{\Xbasis,1,1} \bar{e}_{\Xbasis,1,1} \right)^\downarrow - \Delta
 Y_{\Xbasis,1,1} \bar{e}_{\Xbasis,1,1} ,
 \label{E:y_def_shift4}
\end{equation}
\begin{equation} 
 t \approx \frac{s_\Xbasis \langle Q_{\Xbasis,i,j} E_{\Xbasis,i,j}
 \rangle_{i,j}}{\langle Q_{\Xbasis,i,j} \rangle_{i,j}} ,
 \label{E:t_def_shift4}
\end{equation}
\begin{equation}
 x = \frac{\left( Y_{\Xbasis,1,1} e_{\Xbasis,1,1} \right)^\downarrow - \Delta
 Y_{\Xbasis,1,1} e_{\Xbasis,1,1}}{t}
 \label{E:x_def_shift4}
\end{equation}
 and
\begin{equation}
 {\mathcal W} = \left\{ \frac{\langle Q_{\Xbasis,i',j'} \rangle_{i',j'}
 {\mathcal A}_{\Xbasis,1,i}^\even
 {\mathcal A}_{\Xbasis,1,j}^\even}{s_\Xbasis p_{i\mid\Xbasis}
 p_{j\mid\Xbasis}} \right\}_{i,j=1}^{k_\Xbasis}
 \label{E:W_def_shift4}
\end{equation}
\end{subequations}
 Last but not least, I need to define $w^{(m)}$ and $n^{(m)}$.
 Recall that by following the analysis in Ref.~\cite{CN20}, there is a one-one
 correspondence between a random variable in ${\mathcal W}$ taking the value
 of $\langle Q_{\Xbasis,i',j'} \rangle_{i',j'}
 {\mathcal A}_{\Xbasis,1,i}^\even {\mathcal A}_{\Xbasis,1,j}^\even /
 (s_\Xbasis p_{i\mid\Xbasis} p_{j\mid\Xbasis})$ and an event that a photon
 pulse pair is prepared by Alice (Bob) using intensity $\mu_{\Xbasis,i}$
 ($\mu_{\Xbasis,j}$) both in basis $\Xbasis$ and that the Bell basis
 measurement result announced by Charlie is inconsistent with the photon states
 prepared by Alice and Bob.
 Now let us arrange the $k_\Xbasis^2$ elements in the set ${\mathcal W}$ are
 arranged in descending order as $\{ w^{(1)}, w^{(2)}, \cdots,
 w^{(k_\Xbasis^2)} \}$.  Then, $n^{(i)}$ is the number of Bell basis
 measurement events that corresponds to the value of $w^{(i)} \in
 {\mathcal W}$.

 There is an important subtlety that requires attention.  In almost all cases
 of interest, each summard in Eq.~\eqref{E:finite-size_e_Z11_shift4} consists
 of three terms.  The first two are positive and the third one is negative.
 The sum of the first two terms almost exactly equal to the magnitude of the
 third term.  Hence, truncation error is very serious if one directly use
 Eq.~\eqref{E:finite-size_e_Z11_shift4} to numerically compute $\hat{r}$.
 The solution is to expand each term in powers of $1/D_m$ and/or $1/E_m$
 defined below.  This gives
\begin{align}
 & \hat{r}^2 \nonumber \\
 \approx{}& \frac{y^2 \Width({\mathcal W})^2}{3} \sum_{m=1}^{k_\Xbasis^2}
  \frac{n^{(m)}}{D_m E_m} \left( \frac{1}{D_m^2} + \frac{1}{D_m E_m} +
  \frac{1}{E_m^2} \right) ,
 \label{E:e_Z11_bound4_approx}
\end{align}
 where
\begin{subequations}
\begin{equation}
 D_m = y+ (t - \sum_{i<m} n^{(i)} + 1)x + \min {\mathcal W} + \sum_{i<m}
 n^{(i)} w^{(i)}
 \label{E:D_m_def}
\end{equation}
 and
\begin{align}
 E_m &= y+ (t - \sum_{i<m} n^{(i)} + 1)x + \min {\mathcal W} + \sum_{i<m}
  n^{(i)} w^{(i)} \nonumber \\
 &\quad + n^{(m)} (w^{(m)} - x) \nonumber \\
 &= y+ (t - \sum_{i\le m} n^{(i)} + 1)x + \min {\mathcal W} + \sum_{i\le m}
  n^{(i)} w^{(i)} .
 \label{E:E_m_def}
\end{align}
\end{subequations}
 (Note that only the leading term is kept in
 Eq.~\eqref{E:e_Z11_bound4_approx}.  This is acceptable because the next order
 term is of order of about $1/100$ that of the leading term in all cases of
 practical interest.)

 With all the above discussions, to summarize, the secure key rate $R$ of this
 $\epsilon_\text{cor}$-correct and $\epsilon_\text{sec}$-secure \QKD scheme in
 the finite raw key length situation is lower-bounded by
\begin{widetext}
\begin{subequations}
\begin{align}
 R &\ge \sum_{i,j=1}^{k_\Zbasis} {\mathcal B}_{\Zbasis,i,j} Q_{\Zbasis,i,j} -
  \langle Q_{\Zbasis,i,j} \rangle_{i,j} \left[ \frac{\ln (\chi /
  \epsilon_\text{sec})}{2s_\Zbasis} \right]^{1/2} \Width \left( \left\{
  \frac{{\mathcal B}_{\Zbasis,i,j}}{p_{i\mid\Zbasis} p_{j\mid\Zbasis}}
  \right\}_{i,j=1}^{k_\Zbasis} \right) - p_\Zbasis^2 \left\{
  \vphantom{\frac{\epsilon_\text{sec}}{s_\Zbasis}}
  \langle \mu \exp(-\mu) \rangle_\Zbasis^2 C_{\Zbasis,2}^2 [ 1 - H_2(e_p) ]
  \right. \nonumber \\
 &\quad \left. + f_\text{EC} \langle Q_{\Zbasis,i,j} H_s(E_{\Zbasis,i,j})
  \rangle_{i,j} + \frac{\langle Q_{\Zbasis,i,j}
  \rangle_{i,j}}{\ell_\text{raw}} \left[ 6\log_2 \left(
  \frac{\chi}{\epsilon_\text{sec}} \right) + \log_2 \left(
  \frac{2}{\epsilon_\text{cor}} \right) \right] \right\} .
 \label{E:key_rate_finite-size}
\end{align}
 and
\begin{align}
 R &\ge \sum_{i,j=1}^{k_\Zbasis} {\mathcal B}'_{\Zbasis,i,j} Q_{\Zbasis,i,j} +
  \sum_{i,j=1}^{k_\Xbasis} {\mathcal B}_{\Xbasis,i,j} Q_{\Xbasis,i,j} -
  \langle Q_{\Zbasis,i,j} \rangle_{i,j} \left[ \frac{\ln (\chi /
  \epsilon_\text{sec})}{2s_\Zbasis} \right]^{1/2} \Width \left( \left\{
  \frac{{\mathcal B}'_{\Zbasis,i,j}}{p_{i\mid\Zbasis} p_{j\mid\Zbasis}}
  \right\}_{i,j=1}^{k_\Zbasis} \right) - \langle Q_{\Xbasis,i,j} \rangle_{i,j}
  \times \nonumber \\
 &\qquad \left[ \frac{\ln (\chi / \epsilon_\text{sec})}{2 s_\Xbasis}
  \right]^{1/2} \Width \left( \left\{
  \frac{{\mathcal B}_{\Xbasis,i,j}}{p_{i\mid\Xbasis} p_{j\mid\Xbasis}}
  \right\}_{i,j=1}^{k_\Xbasis} \right) - p_\Zbasis^2 \left\{
  \vphantom{\frac{\epsilon_\text{sec}}{s_\Zbasis}}
  \langle \mu \exp(-\mu) \rangle_\Zbasis^2
  C_{\Xbasis,2}^2 [ 1 - H_2(e_p) ] + f_\text{EC} \langle Q_{\Zbasis,i,j}
  H_s(E_{\Zbasis,i,j}) \rangle_{i,j} \right. \nonumber \\
 &\qquad \left. + \frac{\langle Q_{\Zbasis,i,j}
  \rangle_{i,j}}{\ell_\text{raw}} \left[ 6\log_2 \left(
  \frac{\chi}{\epsilon_\text{sec}} \right) + \log_2 \left(
  \frac{2}{\epsilon_\text{cor}} \right) \right] \right\} .
 \label{E:key_rate_finite-size_alt}
\end{align}
\end{subequations}
\end{widetext}

 I remark that the R.H.S. of the above inequalities implicitly depends on
 $e_{\Xbasis,1,1}$ whose upper bound obeys
 Inequalities~\eqref{E:e_Z11_bound1}, \eqref{E:e_Z11_bound2},
 \eqref{E:e_Z11_bound3}, \eqref{E:e_Z11_bound4}
 and~\eqref{E:e_Z11_bound4_approx}.  Furthermore, when using the key rate in
 Inequality~\eqref{E:key_rate_finite-size}, $\chi = 9 = 4+1+4$ for the first
 three inequalities concerning $e_{\Xbasis,1,1}$ and $\chi = 10$ for the last
 inequality concerning $e_{\Xbasis,1,1}$~\cite{CN20}.  While using the key
 rate in Inequality~\eqref{E:key_rate_finite-size_alt} instead of
 Inequality~\eqref{E:key_rate_finite-size}, $\chi = 9, 10, 10, 11$ for
 Methods~A, B, C and D, respectively.  (For reason for $\chi$ to increase by
 1 except for Method~A by switching the rate formula from
 Inequality~\eqref{E:key_rate_finite-size} to
 Inequality~\eqref{E:key_rate_finite-size_alt} is due to the inclusion of the
 finite-size statistical fluctuations of the lower bound on $Y_{\Xbasis,1,1}$.)

 Compare with the corresponding key rate formula for standard \QKD scheme, the
 most noticeable difference is the presence of additional terms and factors
 involving $C_{\Bbasis,2}^2$ which tend to lower the key rate.  Fortunately,
 $C_{\Bbasis,2}$ roughly scale as $\mu_{\Bbasis,1}^{k_\Bbasis}$ so that in
 practice, these terms and factors are negligible if $k_\Bbasis \gtrsim 2$ to
 $3$.  Finally, I remark that the in the limit of $s_\Zbasis \to +\infty$, the
 key rate formulae in Inequalities~\eqref{E:key_rate_finite-size}
 and~\eqref{E:key_rate_finite-size_alt} are tight in the sense that these
 lower bound are reachable although the condition for attaining them is highly
 unlikely to occur in realistic channels.

\section{Performance Analysis}
\label{Sec:Performance}
 To study the key rate, I use the channel model reported by Ma and
 Razavi in Ref.~\cite{MR12}, which I called the \MR channel.  For this
 channel,
\begin{subequations}
\label{E:Channel_Model}
\begin{equation}
 Q_{\Xbasis,i,j} = 2\beta_{ij}^2 [1+2\beta_{ij}^2-4\beta_{ij}
 I_0(\alpha_{\Xbasis ij}) +I_0(2\alpha_{\Xbasis ij})] ,
 \label{E:Channel_Q_Xij}
\end{equation}
\begin{equation}
 Q_{\Xbasis,i,j} E_{\Xbasis,i,j} = e_0 Q_{\Xbasis,i,j} - 2(e_0 - e_d)
 \beta_{ij}^2 [I_0(2\alpha_{\Xbasis ij}) - 1] ,
 \label{E:Channel_QE_Xij}
\end{equation}
\begin{equation}
 Q_{\Zbasis,i,j} = Q^{(E)}_{ij} + Q^{(C)}_{ij}
 \label{E:Channel_Q_Zij}
\end{equation}
 and
\begin{equation}
 Q_{\Zbasis,i,j} E_{\Zbasis,i,j} = e_d Q^{(C)}_{ij} + (1-e_d) Q^{(E)}_{ij} ,
 \label{E:Channel_QE_Zij}
\end{equation}
 where $I_0(\cdot)$ is the modified Bessel function of the first kind,
\begin{equation}
 \alpha_{\Bbasis ij} = \frac{\sqrt{\eta_\text{A} \mu_{\Bbasis,i} \eta_\text{B}
 \mu_{\Bbasis,i}}}{2} ,
 \label{E:Channel_alpha_def}
\end{equation}
\begin{equation}
 \beta_{ij} = (1-p_d) \exp \left( -\frac{\eta_\text{A} \mu_{\Xbasis,i} +
 \eta_\text{B} \mu_{\Xbasis,j}}{4} \right) ,
 \label{E:Channel_beta_def}
\end{equation}
\begin{equation}
 e_0 = \frac{1}{2} ,
 \label{E:Channel_e0_def}
\end{equation}
\begin{align}
 Q^{(C)}_{ij} &= 2(1-p_d)^2 \exp \left( -\frac{\eta_\text{A} \mu_{\Zbasis,i} +
  \eta_\text{B} \mu_{\Zbasis,j}}{2} \right) \times \nonumber \\
 & \qquad \left[ 1 - (1-p_d) \exp \left( - \frac{\eta_\text{A}
  \mu_{\Zbasis,i}}{2} \right) \right] \times
  \nonumber \\
 & \qquad \left[ 1 - (1-p_d) \exp \left( - \frac{\eta_\text{B}
  \mu_{\Zbasis,j}}{2} \right) \right]
 \label{E:Channel_QC_def}
\end{align}
 and
\begin{align}
 Q^{(E)}_{ij} &= 2p_d(1-p_d)^2 \exp \left( - \frac{\eta_\text{A}
  \mu_{\Zbasis,i} + \eta_\text{B} \mu_{\Zbasis,j}}{2} \right) \times \nonumber
  \\
 & \qquad \left[ I_0(2\alpha_{\Zbasis ij}) - (1-p_d) \exp \left( -
  \frac{\eta_\text{A} \mu_{\Zbasis,i} + \eta_\text{B} \mu_{\Zbasis,j}}{2}
  \right) \right] .
 \label{E:Channel_QE_def}
\end{align}
 Here $e_d$ is the misalignment probability, $p_d$ is the dark count rate per
 detector.  Moreover, $\eta_\text{A}$ ($\eta_\text{B}$) is transmittance of
 the channel between Alice (Bob) and Charlie.  They are given by
\begin{equation}
 \eta_\text{A} = \eta_d 10^{-\eta_\text{att} L_\text{A} / 10}
 \label{E:Channel_eta_d_def}
\end{equation}
\end{subequations}
 and similarly for $\eta_\text{B}$, where $L_\text{A}$ is the length of the
 fiber connecting Alice to Charlie, $\eta_d$ is the detection efficiency of
 a detector, and $\eta_\text{att}$ is the transmission fiber loss constant.

 I remark that the \MR channel model assumes that the partial Bell state
 measurement is performed using linear optics with idea beam and/or
 polarization beam splitters.  It also assumes that all photon detectors are
 identical and that the dead time is ignored.  Moreover, this channel does not
 consider the use of quantum repeater.

\begin{figure}[t]
 \includegraphics[scale=0.9]{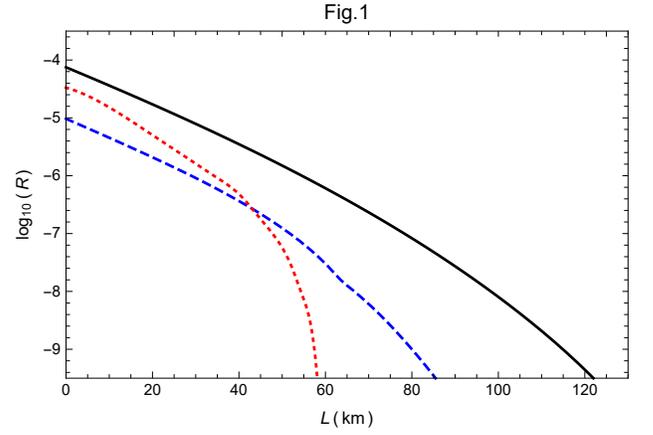}
 \caption{\label{F:Mao_cmp}  Provably secure optimized key rates $R$ as a
  function of distance $L$ between Alice and Bob for $N_t = 10^{10}$ and
  $\epsilon_\text{sec}/\chi = \epsilon_\text{cor} = 10^{-10}$.
  The red dotted curve is the state-of-the-art provably secure key rate
  reported in Ref.~\cite{MZZZZW18}.
  The black solid curve is the rate computed for $(3,2)_G$; and the blue
  dashed curve is rate for $(3,3)_R$.
  The rates for $(k_\Xbasis\ge 3,k_\Zbasis\ge 3)_G$ and $(k_\Xbasis\ge 4,
  k_\Zbasis\ge 4)_R$ is higher than that of $(3,2)_G$ by about 33\%.  But they
  are omitted here as those curves would visually almost overlap with the
  black solid curve in this semi-log plot.
 }
\end{figure}

 The state-of-the-art key rate formula for decoy-state \MDI-\QKD with finite
 raw key length is the one by Mao \emph{et al.} in Ref.~\cite{MZZZZW18}, which
 extended an earlier result by Zhou \emph{et al.} in Ref.~\cite{ZYW16}.  (Note
 that even higher key rates have been reported by Xu
 \emph{et al.}~\cite{XXL14} and Zhou \emph{et al.}~\cite{ZYW16}.  Note however
 that the first work applied brute force optimization as well as the Chernoff
 bound on a much longer raw key.  Its effectiveness in handling short raw key
 length situation is not apparent.  Whereas the second work assumed that the
 statistical fluctuation is Gaussian distributed which is not justified in
 terms of unconditional security.)
 To compare with the provably secure key rate reported in
 Ref.~\cite{MZZZZW18}, I use their settings by putting $e_d = 1.5\%$, $p_d =
 6.02\times 10^{-6}$, $\eta_\text{att} = 0.2$~db/km, $\eta_d = 14.5\%$,
 $f_\text{EC} = 1.16$, and $L_\text{A} = L_\text{B} = L/2$ where $L$ is the
 transmission distance between Alice and Bob.
 For the security parameters, I follow Ref.~\cite{MZZZZW18} by setting
 $\epsilon_\text{sec} / \chi = 10^{-10}$ although a more meaningful way is to
 set $\epsilon_\text{sec}$ divided by the length of the final key to a fixed
 small number~\cite{LCWXZ14}.  Whereas for $\epsilon_\text{cor}$, its value
 has not been specified in Refs.~\cite{MZZZZW18}.  Fortunately,
 Inequality~\eqref{E:key_rate_finite-size} implies that the provably secure
 key rate does not sensitively dependent on $\epsilon_\text{cor}$.  Here I
 simply set it to $10^{-10}$.

 Fig.~\ref{F:Mao_cmp} compares the key rates when the total number of photon
 pulse pairs prepared by Alice and Bob, $N_t \approx \ell_\text{raw} /
 (p_\Zbasis^2 \langle Q_{\Zbasis,i,j} \rangle_{i,j})$ is set to $10^{10}$.
 For each of the curves, the number of photon intensities $k_\Xbasis$
 ($k_\Zbasis)$ used for $\Xbasis$ ($\Zbasis$) are fixed.  The smallest photon
 intensities $\mu_{\Xbasis,k_\Xbasis}$ and $\mu_{\Zbasis,k_\Zbasis}$ are both
 set to be $10^{-6}$.  The optimized key rate is then calculated by varying
 the other $\mu_{\Xbasis,i}$'s, $\mu_{\Zbasis,i}$'s as well as
 $p_{i|\Xbasis}, p_{i|\Zbasis}$ and $p_Z$ by a combination of random sampling
 (with a minimum of $10^7$~samples to a maximum of about $10^9$~samples per
 data point on each curve) and adaptive gradient descend method (that is, the
 step size is adjusted dynamically to speed up the descend).
 For some of the curves, I introduce additional constraints that
 $\mu_{\Xbasis,i} = \mu_{\Zbasis,i}$ so as to reduce the number of different
 photon intensities used.  To aid discussion, I refer to the unconstrained and
 constrained situations by $(k_\Xbasis,k_\Zbasis)_G$ and
 $(k_\Xbasis,k_\Zbasis)_R$, respectively.

 The $R-L$~graphs in Fig.~\ref{F:Mao_cmp} clearly show the advantage of using
 the method in this text in computing the provably secure (optimized) key
 rate.
 The black curve, which is the distance-rate graph of $(3,2)_G$ that uses four
 different photon intensities, is much better than the red one (which also
 uses four different photon intensities) originally reported in
 Ref.~\cite{MZZZZW18}.  In fact, for any distance $L$ between Alice and Bob,
 the key rate of the $(3,2)_G$ method is at least $2.25$~times that of the
 state-of-the-art key rate reported in Ref.~\cite{MZZZZW18}.
 (I also mention on passing that the rate of the black curve is even higher
 than that of the two decoy key rate reported in Ref.~\cite{XXL14} using a
 much longer raw key length of $\ell_\text{raw} = 10^{12}$.)
 Besides, the $(3,2)_G$ method extends the working distance between Alice and
 Bob from slightly less than 60~km to slightly over 130~km.
 The blue dashed curve is the key rate of $(3,3)_R$, which uses the same set
 of three different photon intensities for both preparation bases.  Although
 it uses one less photon intensity, it still outperforms the key rate of the
 red curve when $L \gtrsim 45$~km.

\begin{figure}[t]
 \includegraphics[scale=0.9]{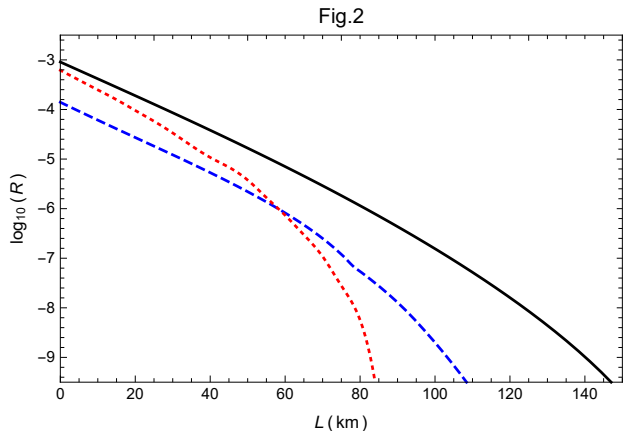}
 \caption{\label{F:Zhou_cmp}  Provably secure optimized key rates $R$ as a
  function of distance $L$ between Alice and Bob for $N_t = 10^9$,
  $\epsilon_\text{sec}/\chi = \epsilon_\text{cor} = 10^{-7}$.
  The red dotted curve is the provably secure key rate reported in
  Ref.~\cite{ZYW16}.
  The black solid curve is the rate computed for $(3,2)_G$; and the blue
  dashed curve is rate for $(3,3)_R$.
  The rates for $(k_\Xbasis\ge 3,k_\Zbasis\ge 3)_G$ are about the same as that
  of $(3,2)_G$ while those of $(k_\Xbasis\ge 4,k_\Zbasis\ge 4)_R$ are a little
  bit lower (higher) than that of $(3,2)_G$ for short (long) transmission
  distances.  They are omitted here to avoid curve jamming.
 }
\end{figure}

 To further illustrate the power of this method, I compare the key rates
 here with the ones obtained in Ref.~\cite{ZYW16} in which they used four
 photon states and the following paramters: $e_d = 1.5\%$, $p_d = 10^{-7}$,
 $\eta_\text{att} = 0.2$~db/km, $\eta_d = 40\%$, $f_\text{EC} = 1.16$,
 $L_\text{A} = L_\text{B} = L/2$ and $\epsilon_\text{sec} / \chi =
 \epsilon_\text{cor} = 10^{-7}$.  The optimized key rates are then found using
 the same method that produces Fig.~\ref{F:Mao_cmp}.
 As shown in Fig.~\ref{F:Zhou_cmp}, the optimized key rate of $(3,2)_G$ (the
 black curve that uses four different photon intensities) is at least 90\%
 higher than those reported in Ref.~\cite{ZYW16}.  Just like the previous
 comparison, the key rate of $(3,3)_R$ which uses only three different photon
 intensities is better than the one reported in Ref.~\cite{ZYW16} when $L
 \gtrsim 60$~km.  Last but not least, the maximum transmission distance
 increases from about 87~km to about 156~km for $(3,2)_G$ and 162~km for
 $(4,3)_G$ (the later not sure in the figure to avoid curve crowding).

\begin{figure}[t]
 \includegraphics[scale=0.9]{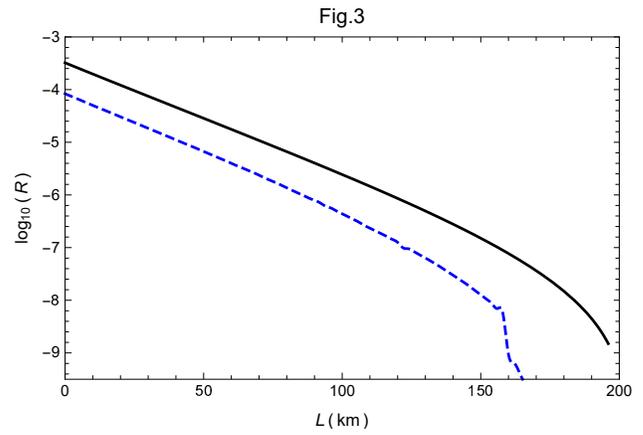}
 \caption{\label{F:kappa_rates} Provably secure optimized key rates $R$ as a
  function of distance $L$ between Alice and Bob for $\ell_\text{raw} =
  10^{10}$ and $\kappa = 10^{-15}$ and $\epsilon_\text{cor} = 10^{-10}$.
  The black solid curve is for $(3,2)_G$ and the blue dashed curve is for
  $(3,3)_R$.
  The rates for $(k_\Xbasis\ge 3,k_\Zbasis\ge 3)_G$ and $(k_\Xbasis\ge 4,
  k_\Zbasis\ge 4)_R$ is higher than that of $(3,2)_G$ by about 18\%.
  These additional curves are omitted here as visually they almost overlap
  with the black solid curve in this semi-log plot.
 }
\end{figure}

 In a sense, instead of $\epsilon_\text{sec} / \chi$, a fairer security
 parameter to use is $\kappa$, namely $\epsilon_\text{sec}$ per number of bits
 of the final key~\cite{LCWXZ14}.
 Fig.~\ref{F:kappa_rates} depicts the $R-L$ curves of various methods when
 $\kappa = 10^{-15}$ and $\epsilon_\text{cor} = 10^{-10}$.  Here instead of
 fixing $N_t$, I keep $\ell_\text{raw} = 10^{10}$ which corresponds to a much
 greater value of $N_t$ in general.
 The blue dashed curve is the rate for $(3,3)_R$ which uses three photon
 intensities.  It already achieves a non-zero key rate at a distance of
 slightly greater than 160~km.  The black curve is the rate for $(3,2)_G$
 which uses four photon intensities.
 It allows Alice and Bob to share a secret key over a distance close to
 200~km.
 This finding makes sense for a larger raw key length $\ell_\text{raw}$
 implies smaller statistical fluctuations in our estimates of various yields
 and error rate, which in turn increase the provably secure key rate and the
 maximum transmission distance.

 Tables~\ref{T:keyrates_epsilon_sec_per_chi}
 and~\ref{T:keyrates_kappa} shows the provably secure optimized key rates
 using various values of $k_\Xbasis$ and $k_\Zbasis$ for the case of fixing
 $\epsilon_\text{sec} / \chi$ and $\kappa$, respectively.  The following
 points can be drawn from these figures and tables.
 First, for the unconstrained photon intensity situation, the optimized key
 rate increases as $k_\Xbasis$ increases.
 For instance, as shown in Table~\ref{T:keyrates_epsilon_sec_per_chi}, the
 key rate of $(4,2)_G$ is at least 39\% higher than that of $(3,2)_G$ by
 fixing $\epsilon_\text{sec}/\chi$.  And from Table~\ref{T:keyrates_kappa},
 the corresponding increase in key rate by fixing $\kappa$ is about 18\% in
 the distance range from 0~km to 150~km.
 (I do not draw these curves in Figs.~\ref{F:Mao_cmp} and~\ref{F:kappa_rates}
 because they almost overlap with the $(3,2)_G$ curve using the same plotting
 scales.)
 Second, for the constrained photon intensity situation, the optimized key
 rate increases as $k_\Xbasis = k_\Zbasis$ increases.  These findings can be
 understood by the fact that the more decoy intensity used, the closer the
 various bounds of yields and error rates are to their actual values.
 Third, the constrained key rate is in general several times lower than the
 corresponding unconstrained one.  So, using the same set of photon
 intensities for the two bases is not a good idea, at least for the \MR
 channel.

\begin{table}[t]
 \centering
 \begin{tabular}{||c|c|c||}
  \hline\hline
  $L$/km & 0 & 50 \\
  \hline
  $(3,2)_G$ & $7.49\times 10^{-5}$ & $1.50\times 10^{-6}$ \\
  \hline
  $(3,3)_R$ & $9.65\times 10^{-6}$ & $1.25\times 10^{-7}$ \\
  \hline
  $(3,3)_G$ & $8.51\times 10^{-5}$ & $1.82\times 10^{-6}$ \\
  \hline
  $(4,2)_G$ & $1.04\times 10^{-4}$ & $2.22\times 10^{-6}$ \\
  \hline
  $(4,3)_G$ & $1.04\times 10^{-4}$ & $2.24\times 10^{-6}$ \\
  \hline
  $(4,4)_R$ & $3.10\times 10^{-5}$ & $3.75\times 10^{-7}$ \\
  \hline
  $(4,4)_G$ & $1.04\times 10^{-4}$ & $2.23\times 10^{-6}$ \\
  \hline\hline
 \end{tabular}
 \caption{Optimized secure key rates for $N_t = 10^{10}$ and
 $\epsilon_\text{sec} / \chi = \epsilon_\text{cor} = 10^{-10}$.
 \label{T:keyrates_epsilon_sec_per_chi}}
\end{table}

 There is an interesting observation that requires in-depth explanation.
 From Table~\ref{T:keyrates_kappa}, for the case of fixing $k_\Xbasis$ and
 $\kappa$, the increase in $R$ due to increase in $k_\Zbasis$ is
 insignificant.
 Moreover, Table~\ref{T:keyrates_epsilon_sec_per_chi} shows that for the case
 fixing $\epsilon_\text{sec} / \chi$ and $k_\Xbasis$, significant increase in
 key rate occurs only when $k_\Xbasis = 3$.
 The reason is that for the \MR channel~\cite{MR12}, it turns out that the key
 rate computed by Inequality~\eqref{E:key_rate_finite-size_alt} is greater
 than that computed by Inequality~\eqref{E:key_rate_finite-size}.
 That is to say, the lower bound of $Y_{\Xbasis,1,1}$ is a better estimate of
 the single photon pair yield that the lower bound of $Y_{\Zbasis,1,1}$.
 Thus, increasing $k_\Zbasis$ only gives a better estimate of
 $Y_{\Zbasis,0,\star}$.  Since I fix the lowest photon intensity to $10^{-6}$,
 which is very close to the vacuum state, the major error in estimating
 $Y_{\Zbasis,0,\star}$ comes from finite-size statistical fluctuation.
 Consequently, by fixing a large enough raw key length $\ell_\text{raw}$, the
 use of more than two photon intensities for the $\Zbasis$ does not improve
 the provably secure key rate in practice.  In other words, the improvement on
 the provably secure key rate by increasing $k_\Zbasis$ alone for the \MR
 channel occurs only when $k_\Zbasis$ is small, say, about 2 to 3 and when
 $\ell_\text{raw}$ is small.
 
 There is a systematic trend that worth reporting.
 For the case of using unconstrained photon intensities, Method~D plus the
 use of $Y_{\Xbasis,1,1}$ to bound the single photon-pair yield gives the
 highest key rate over almost the whole range of distance $L$.  Only when
 close to the maximum transmission distance that the best rate is computed
 using Method~C and $Y_{\Xbasis,1,1}$.
 Whereas for the case of constrained photon intensities, for short transmission
 distance, the best rate is computed using Method~D plus $Y_{\Zbasis,1,1}$.
 For longer transmission distance, the best rate is due to Method~B and
 $Y_{\Xbasis,1,1}$.

\begin{table}[t]
 \centering
 \begin{tabular}{||c|c|c|c|c||}
  \hline\hline
  $L$/km & 0 & 50 & 100 & 150 \\
  \hline
  $(3,2)_G$ & $3.23\times 10^{-4}$ & $2.85\times 10^{-5}$ &
  $2.44\times 10^{-6}$ & $1.51\times 10^{-7}$ \\
  \hline
  $(3,3)_R$ & $8.37\times 10^{-5}$ & $6.67\times 10^{-6}$ &
  $4.33\times 10^{-7}$ & $1.27\times 10^{-8}$ \\
  \hline
  $(3,3)_G$ & $3.23\times 10^{-4}$ & $2.85\times 10^{-5}$ &
  $2.44\times 10^{-6}$ & $1.51\times 10^{-7}$ \\
  \hline
  $(4,2)_G$ & $3.82\times 10^{-4}$ & $3.39\times 10^{-5}$ &
  $2.89\times 10^{-6}$ & $1.78\times 10^{-7}$ \\
  \hline
  $(4,3)_G$ & $3.82\times 10^{-4}$ & $3.39\times 10^{-5}$ &
  $2.89\times 10^{-6}$ & $1.78\times 10^{-7}$ \\
  \hline
  $(4,4)_R$ & $1.70\times 10^{-4}$ & $1.32\times 10^{-5}$ &
  $8.27\times 10^{-7}$ & $2.64\times 10^{-8}$ \\
  \hline
  $(4,4)_G$ & $3.82\times 10^{-4}$ & $3.39\times 10^{-5}$ &
  $2.89\times 10^{-6}$ & $1.78\times 10^{-7}$ \\
  \hline\hline
 \end{tabular}
 \caption{Optimized secure key rates for $\ell_\text{raw} = 10^{10}$,
  $\kappa = 10^{-15}$ and $\epsilon_\text{cor} = 10^{-10}$.
 \label{T:keyrates_kappa}}
\end{table}

\section{Summary}
\label{Sec:Summary}
 In summary, using the BB84 scheme in the \MR channel as an example, I have
 reported a key rate formula for \MDI-\QKD using Possionian photon sources
 through repeated use the inversion of Vandermonde matrix and a McDiarmid-type
 inequality.
 This method gives a provably secure key rate that is at least 2.25~times that
 of the current state-of-the-art result.  It also shows that using five photon
 intensities, more precisely the $(4,2)$-method, gives an additional 18\%
 increase in the key rate for the \MR channel.
 This demonstrates once again the effectiveness of using McDiarmid-type
 inequality in statistical data analysis in physical problems.
 Provided that the photon source is sufficiently close to Possionian,
 
 Note that the Vandermonde matrix inversion technique is rather general.  As
 pointed out in Remark~\ref{Rem:C_property} in the Appendix, by modifying the
 proof of Lemma~\ref{Lem:C_property}, one can show that $C_{3i}
 \ge 0$ if $k$ is even and $C_{3i} < 0$ if $k$ is odd for all $i \ge k$.
 Thus, I can find the lower bound of $Y_{\Bbasis,0,2}$ and $Y_{\Bbasis,2,0}$.
 In other words, I can extend the key rate calculation to the case of
 twin-field~\cite{LYDS18} or phase-matching \MDI-\QKD~\cite{MZZ18}.
 Note further that Inequalities~\eqref{E:various_Y_and_e_bounds} are still
 valid by replacing the ${\mathcal A}_{\Bbasis,j,i}^\even$'s and
 ${\mathcal A}_{\Bbasis,j,i}^\odd$'s by their perturbed expressions through
 standard matrix inversion perturbation as long as the photon sources are
 sufficiently close to Possionian.
 In this regard, the theory developed here also applies to these sources.
 Interested readers are invited to fill in the details.
 Last but not least, it is instructive to extend this work to cover other
 \MDI-\QKD protocols as well as more realistic quantum channels that take dead
 time and imperfect beam splitters into account.

\appendix

\section{Auxiliary Results On Bounds Of Yields And Error Rates}
\label{Sec:results_on_Y_and_e_bounds}
 I begin with the following lemma.
\begin{Lem}
 Let $\mu_1,\mu_2,\cdots,\mu_k$ be $k \ge 2$ distinct real numbers.  Then
 \begin{equation}
  \sum_{i=1}^k \frac{\mu_i^\ell}{\prod_{t\ne i} (\mu_i - \mu_t)} = 0
  \label{E:auxiliary_sum}
 \end{equation}
 for $\ell = 0,1,\cdots,k-2$.
 \label{Lem:sum}
\end{Lem}

\begin{proof}
 Note that the L.H.S. of Eq.~\eqref{E:auxiliary_sum} is a symmetric function
 of $\mu_i$'s.
 Moreover, only its first two terms involve the factor $\mu_1-\mu_2$ in the
 denominator.  In fact, the sum of the these two terms equals
 \begin{displaymath}
  \frac{\mu_1^\ell \prod_{t>2} (\mu_2 - \mu_t) - \mu_2^\ell \prod_{t>2} (\mu_1
  - \mu_t)}{(\mu_1 - \mu_2) \prod_{t>2} [(\mu_1 - \mu_t) (\mu_2 - \mu_t)]} .
 \end{displaymath}
 By applying reminder theorem, I know that the numerator of the above
 expression is divisible by $\mu_1 - \mu_2$.  
 Consequently, the L.H.S. of Eq.~\eqref{E:auxiliary_sum} is a homogeneous
 polynomial of degree $\le \ell - k + 1$.  But as $\ell \le k-2$, this means
 the L.H.S. of Eq.~\eqref{E:auxiliary_sum} must be a constant.  By putting
 $\mu_i = t^i$ for all $i$ and by taking the limit $t\to +\infty$, I conclude
 that this constant is $0$.  This completes the proof.
\end{proof}

 Following the notation in Ref.~\cite{Chau18}, I define
\begin{equation}
 C_{a+1,i} = \frac{(-1)^{k-a} a!}{i!} \sum_{t=1}^k \frac{\mu_t^i
 S_{ta}}{\prod_{\ell\ne t} (\mu_t - \mu_\ell)} ,
 \label{E:C_def}
\end{equation}
 where
\begin{equation}
 S_{ta} = \sideset{}{^{'}}{\sum} \mu_{t_1} \mu_{t_2} \cdots \mu_{t_{k-a-1}}
 \label{E:S_in_def}
\end{equation}
 with the primed sum being over all $t_j$'s $\ne i$ obeying $1 \le t_1 < t_2 <
 \dots < t_{k-a-1} \le k$.
 The following lemma is an extension of a result in Ref.~\cite{Chau18}.
\begin{Lem}
 Let $\mu_1 > \mu_2 > \cdots > \mu_k \ge 0$.
 Suppose $0\le i < k$.  Then
 \begin{equation}
  C_{a+1,i} =
  \begin{cases}
   -1 & \text{if } a = i , \\
   0 & \text{otherwise} .
  \end{cases}
  \label{E:C_value1}
 \end{equation}
 Whereas if $i \ge k$, then
 \begin{equation}
  \begin{cases}
   C_{1i} \ge 0 \text{ and } C_{2i} < 0 & \text{if } k \text{ is even,} \\
   C_{1i} \le 0 \text{ and } C_{2i} > 0 & \text{if } k \text{ is odd.}
  \end{cases}
  \label{E:C_value2}
 \end{equation}
 \label{Lem:C_property}
\end{Lem}
\begin{proof}
 Using the same argument as in the proof of Lemma~\ref{Lem:sum}, I conclude
 that $C_{a+1,i}$ is a homogeneous polynomial of degree $\le i-a$.

 Consider the case of $i\le a$ so that $C_{a+1,i}$ is a constant.  By putting
 $\mu_t = \delta^t$ for all $t$ and then taking the limit $\delta\to 0^+$, it
 is straightforward to check that $C_{a+1,i} = 0$ if $i < a$ and $C_{a+1,i} =
 -1$ if $i = a$.

 It remains to consider the case of $a < i < k$.  I first consider the subcase
 of $a = 0$.  Here $C_{1i}$ contains a common factor of $\prod_{t=1}^k
 \mu_t$, which is of degree $k > i$.  Therefore, I could write $\prod_{t=1}^k
 \mu_t = C_{1i} F$ where $F$ is a homogeneous polynomial of degree $\ge i$.
 As a consequence, either $C_{1i}$ or $F$ contains $\mu_1$ and hence all
 $\mu_t$'s.  Thus, $C_{1i}$ must be a constant for $i < k$.  By setting $\mu_k
 = 0$, I know that $C_{1i} = 0$.

 Next, I consider the subcase of $a = 1$.  Since $i > 1$, from the findings of
 the first subcase, I arrive at
 \begin{align}
  C_{2i} &= \frac{(-1)^{k-1}}{i!} \sum_{t=1}^k \frac{T_1 \mu_t^{i-1} - \left(
   \prod_{\ell=1}^k \mu_\ell \right) \mu_t^{i-2}}{\prod_{\ell\ne t} (\mu_t -
   \mu_\ell)} \nonumber \\
  &= \frac{(-1)^{k-1} T_1}{i!} \sum_{t=1}^k \frac{\mu_t^{i-1}}{\prod_{\ell\ne
   t} (\mu_t - \mu_\ell)} - C_{1,i-1} \nonumber \\
  &= \frac{(-1)^{k-1} T_1}{i!} \sum_{t=1}^k \frac{\mu_t^{i-1}}{\prod_{\ell\ne
   t} (\mu_t - \mu_\ell)} ,
  \label{E:C_a=1}
 \end{align}
 where $T_1$ is the symmetric polynomial
\begin{equation}
 T_1 = \sum_{t=1}^k \mu_1 \mu_2 \cdots \mu_{t-1} \mu_{t+1} \cdots \mu_k .
 \label{E:T1_def}
\end{equation}
 By Lemma~\ref{Lem:sum}, I find that $C_{2i} = 0$ as $i < k$.

 The third subcase I consider is $a = 2$.  As $i > 2$,
 \begin{equation}
  C_{3i} = \frac{2 (-1)^k}{i!} \sum_{t=1}^k \frac{T_2 \mu_t^{i-1} - T_1
   \mu_t^{i-2} + \left( \prod_{\ell=1}^k \mu_\ell \right)
   \mu_t^{i-3}}{\prod_{\ell\ne t} (\mu_t - \mu_\ell)} ,
  \label{E:C_a=2}
 \end{equation}
 where
\begin{equation}
 T_2 = \sideset{}{^{'}}{\sum} \mu_{t_1} \mu_{t_2} \cdots \mu_{t_{k-2}}
 \label{E:T2_def}
\end{equation}
 with the primed sum over all $t_j$'s with $1\le t_1 < t_2 < \dots < t_{k-2}
 \le k$.  By Lemma~\ref{Lem:sum}, I get $C_{3i} = 0$ as $i < k$.

 By induction, the proof of the subcase $a=2$ can be extended to show the
 validity for all $a \le k$ and $a < i < k$.  This shows the validity of
 Eq.~\eqref{E:C_value1}

 The proof of Eq.~\eqref{E:C_value2} can be found in Ref.~\cite{Chau18}.  I
 reproduce here for easy reference.
 By expanding the $1/(\mu_1 - \mu_t)$'s in $C_{ai}$ as a power series of
 $\mu_1$ with all the other $\mu_t$'s fixed, I obtain
 \begin{align}
  C_{1i} &= \frac{(-1)^k}{i!} \left( \prod_{t=1}^k \mu_t \right) \left[
   \mu_1^{i-k} \prod_{r=2}^k \left( 1 + \frac{\mu_r}{\mu_1} +
   \frac{\mu_r^2}{\mu_1^2} + \cdots \right) \right. \nonumber \\
  &\qquad \left. \vphantom{\prod_{r=2}^k \left( 1 + \frac{\mu_r}{\mu_1} +
   \frac{\mu_r^2}{\mu_1^2} \right)}
   + f(\mu_2,\mu_3,\cdots,\mu_k) \right] + \BigOh(\frac{1}{\mu_1})
  \label{E:C_1i_intermediate}
 \end{align}
 for some function $f$ independent of $\mu_1$.
 As $C_{1i}$ is a homogeneous polynomial of degree $\le i$, by equating terms
 in powers of $\mu_1$, I get
 \begin{equation}
  C_{1i} = \frac{(-1)^k}{i!} \left( \prod_{t=1}^k \mu_t \right)
  \sum_{\substack{t_1+t_2+\dots+t_k = i-k, \\ t_1,t_2,\cdots,t_k \ge 0}}
   \mu_1^{t_1} \mu_2^{t_2} \cdots \mu_k^{t_k}
  \label{E:C_1i_explicit}
 \end{equation}
 for all $i\ge k$.  As all $\mu_t$'s are non-negative, $C_{1i} \ge 0$ if
 $k$ is even and $C_{1i} \le 0$ if $k$ is odd.

 By the same argument, I expand all the $1/(\mu_1 - \mu_t)$ terms in $C_{2i}$
 in powers of $\mu_1$ to get
 \begin{align}
  C_{2i} &= \frac{(-1)^{k-1} T_1}{i!}
   \!\!\sum_{\substack{t_1+t_2+\cdots+t_k = i-k
   \\ t_1 > 0, t_2,\cdots,t_k \ge 0}}
   \mu_1^{t_1} \mu_2^{t_2} \cdots \mu_k^{t_k} \nonumber \\
  & \quad + f'(\mu_2,\cdots,\mu_k)
  \label{E:C_2j_intermediate}
 \end{align}
 for some function $f'$ independent of $\mu_1$.  By recursively expanding
 Eq.~\eqref{E:C_def} in powers of $\mu_2$ but with $\mu_1$ set to $0$, and
 then in powers of $\mu_3$ with $\mu_1,\mu_2$ set to $0$ and so on, I conclude
 that whenever $i \ge k$, then $C_{2i} < 0$ if $k$ is even and $C_{2i} > 0$ if
 $k$ is odd.
 This completes the proof.
\end{proof}

\begin{Rem}
 By the same technique of expanding each factor of $1/(\mu_1 - \mu_t)$ 
 in $C_{a+1,i}$ in powers of $\mu_1$, it is straightforward to show that if
 $i \ge k$ and $j\ge 1$, then $C_{2j+1,i} \ge 0$ and $C_{2j,i} \le 0$ provided
 that $k$ is even.  And $C_{2j+1,i} \le 0$ and $C_{2j,i} \ge 0$ provided that
 $k$ is odd.
 \label{Rem:C_property}
\end{Rem}

 The following theorem is an extension of a similar result reported in
 Ref.~\cite{Chau18} by means of an explicit expression of the inverse of a
 Vandermonde matrix.

\begin{Thrm}
 Let $\mu_1 > \mu_2 > \cdots > \mu_k \ge 0$ and $\tilde{\mu}_1 > \tilde{\mu}_2 >
 \cdots > \tilde{\mu}_{\tilde{k}} \ge 0$.  Suppose
 \begin{equation}
  \sum_{a,b=0}^{+\infty} \frac{\mu_i^a}{a!} \frac{\tilde{\mu}_j^b}{b!} A_{ab}
  \equiv \sum_{a,b=0}^{+\infty} M_{a+1,i} \tilde{M}_{b+1,j} A_{ab} =
  B_{ij}
  \label{E:A_B_relation}
 \end{equation}
 for all $i = 1,2,\cdots,k$ and $j = 1,2,\cdots,\tilde{k}$.
 Then,
 \begin{align}
  A_{ab} &= \sum_{i=1}^k \sum_{j=1}^{\tilde{k}} \left( M^{-1} \right)_{a+1,i}
   \left( \tilde{M}^{-1} \right)_{b+1,j} B_{ij} \nonumber \\
  & \quad + \sum_{I=k}^{+\infty} C_{a+1,I} A_{Ib} +
   \sum_{J=\tilde{k}}^{+\infty} \tilde{C}_{b+1,J} A_{aJ} \nonumber \\
  & \quad - \sum_{I=k}^{+\infty} \sum_{J=\tilde{k}}^{+\infty} C_{a+1,I}
   \tilde{C}_{b+1,J} A_{IJ}
  \label{E:explicit_A_relation}
 \end{align}
 for all $a = 0,1,\cdots,k-1$ and $b = 0,1,\cdots,\tilde{k}-1$.
 Here
 \begin{equation}
  \left( M^{-1} \right)_{a+1,i} = \frac{(-1)^{k-a-1} a! S_{ia}}{\prod_{t\ne i}
  (\mu_i - \mu_t)}
  \label{E:M_inverse}
 \end{equation}
 and similarly for $\left( \tilde{M}^{-1}
 \right)_{b+1,j}$.
 \label{Thrm:double_inversion}
\end{Thrm}

\begin{proof}
 Note that for any fixed $a = 0,1,\cdots,k-1$ and $b = 0,1,\cdots,
 \tilde{k}-1$,
 \begin{equation}
  \sum_{b=0}^{+\infty} \frac{\tilde{\mu}_{j}^b}{b!} A_{ab} = \sum_{i=1}^k
  \left( M^{-1} \right)_{a+1,i} \left( B_{ij} - \sum_{b=0}^{+\infty}
  \sum_{I=k}^{+\infty} \frac{\mu_i^I}{I!} \frac{\tilde{\mu}_j^b}{b!} A_{Ib}
  \right) .
  \label{E:inversion_intermediate1}
 \end{equation}
 Here $M^{-1}$ is the inverse of the $k\times k$ matrix
 $(M_{a+1,i})_{a+1,i=1}^k$.  From Ref.~\cite{Chau18}, the matrix elements of
 $M^{-1}$ are related to inverse of certain Vandermonde matrix and are given
 by the expression immediately after Eq.~\eqref{E:explicit_A_relation}.  From
 Lemma~\ref{Lem:C_property}, Eq.~\eqref{E:inversion_intermediate1} can be
 rewritten as
 \begin{equation}
  \sum_{b=0}^{+\infty} \frac{\tilde{\mu}_j^b}{b!} A_{ab} = \sum_{a=1}^k \left(
  M^{-1} \right)_{a+1,i} B_{ij} + \sum_{b=0}^{+\infty} \sum_{I=k}^{+\infty}
  \frac{\tilde{\mu}_j^b}{b!} C_{a+1,I} A_{Ib} .
  \label{E:inversion_intermediate2}
 \end{equation}
 By repeating the above procedure again, I find that for any fixed $a=0,1,
 \cdots,k-1$ and $b = 0,1,\cdots,\tilde{k}-1$,
 \begin{align}
  A_{ab} &= \sum_{i=1}^k \sum_{j=1}^{\tilde{k}} \left( M^{-1} \right)_{a+1,i}
   \left( \tilde{M}^{-1} \right)_{b+1,j} B_{ij} \nonumber \\
  &\quad - \sum_{I=k}^{+\infty} \sum_{\tilde{t}=0}^{+\infty} C_{a+1,I}
   \tilde{C}_{b+1,\tilde{t}} A_{I\tilde{t}} + \sum_{J=\tilde{k}}^{+\infty}
   \tilde{C}_{b+1,J} A_{aJ} .
  \label{E:inversion_intermediate3}
 \end{align}
 Here the $\tilde{k}\times\tilde{k}$ matrix $\tilde{C}$ is defined in the
 exactly the same as the $k\times k$ matrix $C$ except that the $k$ and
 $\mu_t$'s variables are replaced by $\tilde{k}$ and the corresponding
 $\tilde{\mu}_t$'s.
 Substituting Eq.~\eqref{E:C_value1} into the above equation gives
 Eq.~\eqref{E:explicit_A_relation}.
\end{proof}

 Applying Lemma~\ref{Lem:C_property} and Theorem~\ref{Thrm:double_inversion},
 in particular, the Inequality~\eqref{E:C_value2}, I arrive at the following
 two Corollaries.

\begin{Cor}
 Suppose the conditions stated in Theorem~\ref{Thrm:double_inversion} are
 satisfied.  Suppose further that $A_{ab} = 0$ for all $b > 0$ and $a\ge 0$;
 and $A_{a0} \in [0,1]$ for all $a$.  Then
 \begin{equation}
  A_{00} \ge \sum_{i=1}^k \left( M^{-1} \right)_{1i} B_{i0} .
  \label{E:A_0*_lower_bound}
 \end{equation}
 \label{Cor:Y0*}
\end{Cor}

\begin{Cor}
 Suppose the conditions stated in Theorem~\ref{Thrm:double_inversion} are
 satisfied.  Suppose further that $A_{ab} \in [0,1]$ for all $a,b$.  Then
 \begin{subequations}
 \begin{equation}
  A_{00} \ge \sum_{i=1}^k \sum_{j=1}^{\tilde{k}} \left( M^{-1} \right)_{1i}
  \left( \tilde{M}^{-1} \right)_{1j} B_{ij} - \sum_{I=k}^{+\infty}
  \sum_{J=\tilde{k}}^{+\infty} C_{1I} \tilde{C}_{1J}
  \label{E:A00_lower_bound}
 \end{equation}
 and
 \begin{equation}
  A_{11} \le \sum_{i=1}^k \sum_{j=1}^{\tilde{k}} \left( M^{-1} \right)_{2i}
  \left( \tilde{M}^{-1} \right)_{2j} B_{ij}
  \label{E:A11_upper_bound}
 \end{equation}
 provided both $k$ and $\tilde{k}$ are even.  Furthermore,
 \begin{equation}
  A_{11} \ge \sum_{i=1}^k \sum_{j=1}^{\tilde{k}} \left( M^{-1} \right)_{2i}
  \left( \tilde{M}^{-1} \right)_{2j} B_{ij} - \sum_{I=k}^{+\infty}
  \sum_{J=\tilde{k}}^{+\infty} C_{2I} \tilde{C}_{2J}
  \label{E:A11_lower_bound}
 \end{equation}
 \end{subequations}
 if both $k$ and $\tilde{k}$ are odd.
 \label{Cor:bounds_on_Yxx}
\end{Cor}

\begin{Rem}
 Clearly, each of the bounds in the above Corollary are tight.  Although the
 conditions for attaining the bound in Inequality~\eqref{E:A11_upper_bound}
 are not compatible with those for attaining the bounds in
 Inequalities~\eqref{E:A00_lower_bound} and~\eqref{E:A11_lower_bound}, the
 way I use these inequalities in Secs.~\ref{Sec:Rate}
 and~\ref{Sec:Finite_Size} ensures that it is possible to attaining all
 these bounds in the key rate formula.
 \label{Rem:tightness}
\end{Rem}

\begin{acknowledgments}
 This work is supported by the RGC grant 17302019 of the Hong Kong SAR
 Government.
\end{acknowledgments}

\bibliographystyle{apsrev4-1}

\bibliography{qc76.2}

\end{document}